\documentclass[11pt, leqno, letterpaper]{amsart}

\usepackage{graphicx}    % needed for including graphics e.g. EPS, PS

\usepackage[margin=1.3in]{geometry}

\usepackage{amsmath,amsthm,amsfonts,amssymb}

\newtheorem{theorem}{Theorem}

\newtheorem{corollary}[theorem]{Corollary}
\newtheorem{definition}[theorem]{Definition}

\newtheorem{lemma}[theorem]{Lemma}

\newtheorem{proposition}[theorem]{Proposition}

\newtheorem{remark}[theorem]{Remark}

\begin{document}
 % Start your text

\title[Short Maturity Asian Options]
{Short Maturity Asian Options in Local Volatility Models}

\author{Dan Pirjol}
\email{
dpirjol@gmail.com}

\author{Lingjiong Zhu}
\address
{Department of Mathematics\newline
\indent Florida State University\newline
\indent 1017 Academic Way\newline
\indent Tallahassee, FL-32306\newline
\indent United States of America}
\email{
zhu@math.fsu.edu}

\date{19 September 2016}

\subjclass[2010]{91G20,91G80,60F10}%Derivative securities, financial applications of other theories, large deviations
\keywords{Asian options, short maturity, local volatility, large deviations, variational problem.}

\begin{abstract}
We present a rigorous study of the short maturity asymptotics 
for Asian options with continuous-time averaging, under the assumption  that 
the underlying asset follows a local volatility model. 
The asymptotics for out-of-the-money, in-the-money, and at-the-money cases are 
derived, considering both fixed strike and floating strike Asian options. 
The asymptotics for the out-of-the-money case involves a non-trivial 
variational problem which is solved completely. 
We present an analytical approximation for Asian options prices, and
demonstrate good numerical agreement of the asymptotic results with the
results of Monte Carlo simulations and benchmark test cases
in the Black-Scholes model for option parameters relevant in practical 
applications.
\end{abstract}

\maketitle

\section{Introduction}

Asymptotics of the option prices, and implied volatility for short maturity European options have
been extensively studied in the literature, see e.g. 
\cite{BBF,GHLOW,GW,Durrleman2010,CCLMN} for 
the local volatility models,
\cite{FF2012,Tankov,AL2013,Durrleman2008,RR} for the exponential L\'{e}vy models and
\cite{BBFII,Henry,FJL,Feng,FengII,Forde,Alos} for the stochastic volatility models and 
\cite{Gao,MN2011} for model-free approaches.
To the best of our knowledge, the short maturity Asian options have been much 
less studied. Unlike the European options, the Asian options do not have a 
simple closed-form formula even in the Black-Scholes model. That is why in 
financial industry, the Asian options are quoted by price rather than implied 
volatility. 
In this paper, we study the short maturity asymptotics for the price
of the Asian options under the assumption that the underlying asset price
follows a local volatility model. We obtain analytical results for the 
short maturity asymptotics of Asian options in the local volatility model, and
more explicit results in the case of the Black-Scholes model. 
We define and study the implied volatility of an Asian option in the short
maturity limit. 

Let us assume that the stock price follows a local volatility model:
\begin{equation}\label{dynamics}
dS_{t}=(r-q)S_{t}dt+\sigma(S_{t})S_{t}dW_{t},
\qquad S_{0}>0,
\end{equation}
where $W_{t}$ is a standard Brownian motion, $r\geq 0$ is the risk-free rate, 
$q\geq 0$ is the continuous dividend yield, $\sigma(\cdot)$ is the local 
volatility and the log-stock price process $X_t=\log S_t$ satisfies
\begin{equation*}
dX_t = 
\left( r - q - \frac12 \sigma^2(e^{X_t}) \right) dt +
\sigma(e^{X_t}) dW_t\,.
\end{equation*}
We assume that the local volatility function $\sigma(\cdot)$ satisfies
\begin{align}
&0<\underline{\sigma}\leq\sigma(\cdot)\leq\overline{\sigma}<\infty,\label{assumpI}
\\
&|\sigma(e^x)-\sigma(e^y)|\leq M|x-y|^{\alpha},\label{assumpII}
\end{align}
for some fixed $M,\alpha>0$ for any $x,y$ and $0<\underline{\sigma}<\overline{\sigma}<\infty$ are some fixed constants.

The price of the Asian call and put options with maturity $T$ and strike $K$ 
are given by
\begin{align}
&C(T):=e^{-rT}\mathbb{E}\left[\left(\frac{1}{T}\int_{0}^{T}S_{t}dt-K\right)^{+}\right],
\\
&P(T):=e^{-rT}\mathbb{E}\left[\left(K-\frac{1}{T}\int_{0}^{T}S_{t}dt\right)^{+}\right],
\end{align}
where $C(T)$ and $P(T)$ emphasize the dependence on the maturity $T$.

When $S_{0}<K$, the call option is out-of-the money and $C(T)\rightarrow 0$ as $T\rightarrow 0$
and when $S_{0}>K$, the put option is out-of-the money and $P(T)\rightarrow 0$ as $T\rightarrow 0$.
When $S_{0}=K$, i.e. at-the-money, both $C(T)$ and $P(T)$ tend to zero
as $T\rightarrow 0$. 
We are interested to study the first-order approximations
of the call and put prices as $T\rightarrow 0$. 
It turns out that the asymptotics for out-of-the-money case are governed by the rare events (large deviations)
and the asymptotics for at-the-money case are governed by the fluctuations about the typical events (Gaussian fluctuations).

There are numerous works in the mathematical finance literature studying the 
pricing of Asian options. The pricing under the Black-Scholes model has been 
studied in \cite{GY,CS,DufresneLaguerre,Linetsky}, using a relation between 
the distributional property of the time-integral of the geometric Brownian 
motion and Bessel processes. See \cite{DufresneReview} for an overview, and 
\cite{FMW} for a comparison with alternative simulation methods, including the 
Monte Carlo approach. A popular method, which has the advantage of wider 
applicability to other models, is the PDE approach \cite{Ingersoll,RogersShi,Vecer,VecerXu}.
The resulting PDE can be solved either numerically \cite{Vecer,VecerXu}, 
or can be used to derive analytical approximation formulae
using asymptotic expansion methods. Such results have been obtained in
\cite{FPP2013} in the local volatility model, and in \cite{Dassios} in the
CEV model. The paper \cite{FPP2013} used heat kernel expansion methods and 
developed approximate formulae expressed in terms of elementary functions for 
the density, the price and the Greeks of path dependent options of Asian style.
Asymptotic expansion leading to analytical approximations with error bounds
for Asian options have been obtained also using Malliavin calculus in 
\cite{Shiraya,GobetMiri}.

Most of the literature on Asian options in the Black-Scholes model, 
i.e., $\sigma(\cdot)\equiv\sigma$, \cite{GY,CS,Linetsky} exploits the 
well-known result \cite{DufresneReview,DGY} that the integral of the geometric 
Brownian motion 
$\int_{0}^{T}e^{(r-q-\frac{1}{2}\sigma^{2})t+\sigma W_{t}}dt$ has the same 
distribution as $X_{t}$, where
\begin{equation}\label{formI}
dX_{t}=\left[(r-q)+\frac{1}{X_{t}}\right]X_{t}dt+\sigma X_{t}dB_{t},\qquad X_{0}=0,
\end{equation}
where $B_{t}$ is a standard Brownian motion. 
The price of the Asian call and put options can thus be computed as
\begin{equation}\label{formII}
C(T)=e^{-rT}\frac{1}{T}\mathbb{E}[(S_{0}X_{T}-TK)^{+}],
\qquad
P(T)=e^{-rT}\frac{1}{T}\mathbb{E}[(TK-S_{0}X_{T})^{+}].
\end{equation}
For the readers who are familiar with the large deviations for small time 
diffusion processes, one may na\"{i}vely believe that the small time 
asymptotics of $X_{T}$ as $T\rightarrow 0$ is comparable with the SDE without
the drift term, i.e., $d\tilde{X}_{t}=\sigma\tilde{X}_{t}d\tilde{B}_{t}$ 
and hence the asymptotics for the short maturity Asian options for the 
Black-Scholes model are the same as for their European counterpart. 
We will show in this paper that this is indeed \emph{not} the case. 
Intuitively, when $\tilde{X}_{t}$ starts
at zero at time zero then $\tilde{X}_{t}$ remains zero for any time $t$. 
Even though $X_{t}$ also starts at zero at time zero, it is kicked to a 
positive value immediately after the time zero. In that respect, the two processes
are not absolutely continuous with respect to each other. Therefore, 
one cannot use the Girsanov theorem to kill the drift term and claim that 
$X_{t}$ process has the same small time asymptotics as a geometric Brownian motion.
Also, the formulas \eqref{formI}, \eqref{formII} are valid only for the 
Black-Scholes model, and this approach does not shed much insight into the more 
general local volatility case. 

In this paper, we will use large deviations theory for small time diffusion processes. The key observation
is that one can apply the contraction principle to get the corresponding large 
deviations for the small time
arithmetic average of the diffusion process, and hence obtain rigorously the asymptotic behavior for the out-of-the-money
Asian call and put options. The asymptotic exponent is given as the rate function from the large deviation principle,
which itself is a complicated and not-so-obvious variational problem. We will manage to solve this variational problem
completely and give a semi-analytical solution in the end. The asymptotics for in-the-money case follows
easily by the put-call parity. We will also obtain the asymptotics 
for at-the-money short maturity Asian options. Unlike the out-of-the-money case, the asymptotics
for at-the-money short maturity has Gaussian fluctuations. 

Most of the existing methods for pricing Asian options are numerically less 
efficient in the limit of small maturities and small volatilities.
For the case of the Asian options under the Black-Scholes model this has been 
noted in the Geman-Yor method \cite{GY,CS}, where the inversion of a Laplace 
transform requires special care for small maturities 
\cite{Shaw,FMW,DufresneLaguerre}. A similar issue appears in the spectral
method \cite{Linetsky}. This issue is not present in methods based on 
asymptotic expansions \cite{FPP2013} which perform well under small 
maturities and volatility conditions. 
The small time expansion presented in this paper is of practical interest as 
it complements some of the alternative approaches in a region where their 
numerical performance is less efficient. 

The paper is organized as follows. In Section \ref{MainSection}, we present asymptotics
for out-of-the-money (OTM), in-the-money (ITM) and at-the-money (ATM) Asian options in a local volatility model
for short maturities. The asymptotics for OTM Asian options involve a not-so-trivial variational problem,
whose solution will be given in Section \ref{Sec:3}, which has a more explicit expression
in the case of Black-Scholes models. The implied volatility and numerical tests will
be discussed in Section \ref{sec:4}. The asymptotics for short maturity floating strike Asian options
will be provided in Section \ref{sec:floating}. Finally, the proofs will be 
given in Section \ref{sec:appendix}.

\section{Asymptotics for Short Maturity Asian Options}\label{MainSection}

Let us recall that the stock price follows a local volatility model:
\begin{equation}
dS_{t}=(r-q)S_{t}dt+\sigma(S_{t})S_{t}dW_{t},
\qquad S_{0}>0,
\end{equation}
where $W_{t}$ is a standard Brownian motion, $r\geq 0$ is the risk-free rate, $q\geq 0$ is the continuous
dividend yield, $\sigma(\cdot)$ is the local volatility satisfying \eqref{assumpI} and \eqref{assumpII}.
We are interested in the short maturity limits, i.e., the asymptotics as $T\rightarrow 0$.

\subsection{Out-of-the-Money and In-the-Money Asian Options}
\label{OTMSection}

We denote the expectation of the averaged asset price in the risk-neutral 
measure as
\begin{equation}\label{Adef}
A(T) := \frac{1}{T} \int_0^T \mathbb{E}[S_t] dt = 
S_0 \frac{1}{(r-q)T} (e^{(r-q)T}-1)\,,
\end{equation}
for $r-q\neq 0$ and $A(T):=S_{0}$ for $r-q=0$,
When $K>A(T)$, the call Asian option is out-of-the-money and 
$C(T)\rightarrow 0$ as $T\rightarrow 0$.
When $A(T)>K$, the put Asian option is out-of-the-money and 
$P(T)\rightarrow 0$ as $T\rightarrow 0$.

\begin{remark}\label{PCParity}
The prices of call and put Asian options are related by put-call parity as
\begin{eqnarray}\label{PCparity}
C(K,T) - P(K,T) = e^{-rT} (A(T) - K)\,.
\end{eqnarray}
\end{remark}

Notice that as $T\rightarrow 0$, $A(T)=S_{0}+O(T)$. 
Therefore, for the small maturity regime, 
the call Asian option is out-of-the-money
if and only if $K>S_{0}$ etc.
And for the rest of the paper, 
the call Asian option is said to be out-of-the-money (resp. in-the-money) if $K>S_{0}$ (resp. $K<S_{0}$), 
and the put Asian option is said to be out-of-the-money (resp. in-the-money) if $K<S_{0}$ (resp. $K>S_{0}$),
and finally they are said to be at-the-money if $K=S_{0}$.

\subsection{Short maturity out-of-the-money Asian options}
We will use large deviations theory to compute the leading order approximation 
at $T\to 0$ for the price of the out-of-the-money Asian options.
\begin{theorem}\label{MainThm}
Assume that \eqref{assumpI} and \eqref{assumpII} both hold.

(i) For out-of-the-money call Asian options, i.e., $K > S_{0}$, 
\begin{equation}
C(T)=e^{-\frac{1}{T}\mathcal{I}(K,S_{0})+o(\frac{1}{T})},\qquad\text{as $T\rightarrow 0$}.
\end{equation}

(ii) For out-of-the-money put Asian options, i.e., $K < S_{0}$, 
\begin{equation}
P(T)
=e^{-\frac{1}{T}\mathcal{I}(K,S_{0})+o(\frac{1}{T})},\qquad\text{as $T\rightarrow 0$}.
\end{equation}
where for any $S_{0},K>0$,
\begin{equation}\label{IRate}
\mathcal{I}(K,S_{0}):=\inf_{
\substack{\int_{0}^{1}e^{g(t)}dt=K
\\
g(0)=\log S_{0}, g\in\mathcal{AC}[0,1]}}\frac{1}{2}\int_{0}^{1}\left(\frac{g'(t)}{\sigma(e^{g(t)})}\right)^{2}dt,
\end{equation}
where $\mathcal{AC}[0,1]$ is the space of absolutely continuous functions on $[0,1]$.

The variational problem in \eqref{IRate} will be solved completely in 
Proposition \ref{VarProp}.
\end{theorem}

\begin{remark}
The small maturity asymptotics given by Theorem~\ref{MainThm} and the rate function
$\mathcal{I}(K,S_0)$ are independent of the interest rate $r$ and 
dividend yield $q$. These quantities contribute only to subleading order in the
$T\to 0$ expansion. This is analogous to the well-known BBF result for the
small maturity European options in the local volatility model \cite{BBF},
which is also independent of $r,q$.
\end{remark}

\begin{remark}
For at-the-money case, i.e. $K = S_{0}$, 
by letting $g(t)\equiv 0$ in \eqref{IRate}, we see that $\mathcal{I}(K,S_{0})=0$.
Indeed, Theorem \ref{ATMThm} will give more precise asymptotics for 
at-the-money short maturity Asian options.
\end{remark}

The asymptotics for short-maturity in-the-money Asian call and put options 
can be obtained as a corollary of the results for out-of-the-money case in 
Theorem \ref{MainThm} by a simple application of the put-call parity.

\begin{corollary}\label{InMoneyCor}
Assume that \eqref{assumpI} and \eqref{assumpII} both hold.

(i) For in-the-money call Asian options, i.e., $K < S_{0}$, 
\begin{equation}
C(T)=S_{0}-K-\frac{1}{2}(r+q)S_{0}T+KrT+O(T^{2}),\qquad\text{as $T\rightarrow 0$}.
\end{equation}

(ii) For in-the-money put Asian options, i.e., $K > S_{0}$, 
\begin{equation}
P(T)=K-S_{0}+\frac{1}{2}(r+q)S_{0}T-KrT+O(T^{2}),\qquad\text{as $T\rightarrow 0$}.
\end{equation}
\end{corollary}

\subsection{At-the-Money Asian Options}\label{ATMSection}

When $K = S_{0}$, the Asian call and put options are at-the-money. 
We have the following result:

\begin{theorem}\label{ATMThm}
Assume that the function $\sigma(s)s$ and $\sigma(s)$ are uniformly Lipschitz, 
i.e., there exist $\alpha,\beta>0$, such that for any $x,y\geq 0$, 
\begin{equation}
|\sigma(x)x-\sigma(y)y|\leq\alpha|x-y|,\qquad|\sigma(x)-\sigma(y)|\leq\beta|x-y|.
\end{equation}

(i) When $K=S_{0}$, as $T\rightarrow 0$,
\begin{equation}
C(T)=\sigma(S_{0})S_{0}\frac{\sqrt{T}}{\sqrt{6\pi}}+O(T).
\end{equation}

(ii) When $K=S_{0}$, as $T\rightarrow 0$,
\begin{equation}
P(T)=\sigma(S_{0})S_{0}\frac{\sqrt{T}}{\sqrt{6\pi}}+O(T).
\end{equation}
\end{theorem}

\begin{remark}
Comparing Theorem \ref{ATMThm} with Theorem \ref{MainThm}, we see that the asymptotics for out-of-the-money 
short maturity Asian options are governed by the rare events (large deviations), while the asymptotics
for at-the-money short maturity Asian options are governed by the fluctuations about the typical events (Gaussian fluctuations).
\end{remark}

\section{Variational Problem for Short-Maturity Asymptotics for Asian Options}
\label{Sec:3}

We present in this section the solution of the variational problem for
the short-time asymptotics of the out-of-the-money Asian options given by 
Theorem~\ref{MainThm}. The solution is given by the following result.

\begin{proposition}\label{VarProp}
The rate function $\mathcal{I}(K,S_0)$ appearing in Theorem~\ref{MainThm}
is given by
\begin{equation}\label{Iresult}
\mathcal{I}(K, S_0)=
\begin{cases}
\frac12 F^{(+)}(h_1) G^{(+)}(h_1) & K \leq S_{0}\,, \\
\frac12 F^{(-)}(f_1) G^{(-)}(f_1) & K \geq S_{0}\,.
\end{cases}
\end{equation}
The two cases are as follows:

(i) $K \leq S_0$. $h_1 \geq 0$ is the solution of the equation
\begin{equation}\label{eqh1}
\frac{K}{S_0} - e^{-h_1} = \frac{G^{(+)}(h_1)}{F^{(+)}(h_1)},
\end{equation}
with
\begin{align}
G^{(+)}(h_1) &= \int_{0}^{h_1} \frac{1}{\sigma( S_0 e^{-y})} 
     \sqrt{e^{-y} - e^{-h_1}}dy, \\
F^{(+)}(h_1) &= \int_{0}^{h_1} \frac{1}{\sigma( S_0 e^{-y})} 
\frac{1}{\sqrt{e^{-y} - e^{-h_1}}}dy \,.
\end{align}

(ii) $K \geq S_0$. $f_1 \geq  0$ is given by the solution of the equation
\begin{equation}\label{eqf1}
e^{f_1} - \frac{K}{S_0} = \frac{G^{(-)}(f_1)}{F^{(-)}(f_1)},
\end{equation}
with
\begin{align}
G^{(-)}(f_1)
&=\int_{0}^{f_1} \frac{1}{\sigma(S_0 e^{y})} 
\sqrt{e^{f_1} - e^{y}}dy,\label{10}
\\
F^{(-)}(f_1)&=\int_{0}^{f_1} \frac{1}{\sigma(S_0 e^{y})} 
\frac{1}{\sqrt{e^{f_1} -  e^{y}}}dy \,.\label{11}
\end{align}
\end{proposition}

%%%%%%%%%%%%%%%%%%%%%%%%%%%%%%%%%%%%%%%%%%%%%%%%%%%%%%%%%%%%%%%%%%%%%%%%%%%%%%%%%%%%%%%%%%%%

We would like to further study the properties of the rate function $\mathcal{I}(K,S_{0})$.
In particular, we will show that for a general local volatility model,
the rate function $\mathcal{I}(K,S_{0})$ is continuous in $K$
and it is increasing in $K$ for $K>S_{0}$ and decreasing in $K$ for $K<S_{0}$. 
This is based on an alternative representation of the rate function.

\begin{proposition}\label{prop:1}

(i) For $K > S_0$ the rate function $\mathcal{I}(K, S_0)$ is given by
\begin{equation}\label{J1}
\mathcal{I}(K,S_0) = \inf_{\varphi > K/S_0}
\frac12 \frac{(\mathcal{G}^{(-)}(\varphi))^2}{\varphi - \frac{K}{S_0}},
\end{equation}
where we denoted
\begin{equation}
\mathcal{G}^{(-)}(\varphi) = \int_1^{\varphi} \frac{\sqrt{\varphi-z}}{z\sigma(S_0 z)} dz
\,,\quad
\varphi \geq 1\,.
\end{equation}

(ii) for $K < S_0$ the rate function $\mathcal{I}(K, S_0)$ is given by
\begin{equation}\label{J2}
\mathcal{I}(K,S_0) = \inf_{0 < \chi < K/S_0}
\frac12 \frac{(\mathcal{G}^{(+)}(\chi))^2}{\frac{K}{S_0} - \chi},
\end{equation}
where we denoted
\begin{equation}
\mathcal{G}^{(+)}(\chi) = \int_{\chi}^1 \frac{\sqrt{z-\chi}}{z\sigma(S_0 z)} dz\,,\quad
0 < \chi \leq 1\,.
\end{equation}
\end{proposition}

\begin{remark}\label{CtsRemark}
In the formulation in Proposition~\ref{prop:1}, we can see that the rate functions
given in \eqref{J1} and \eqref{J2} are indeed continuous in the parameter $K$.
\end{remark}

\begin{figure}[t]
\centering
\includegraphics[width=2.7in]{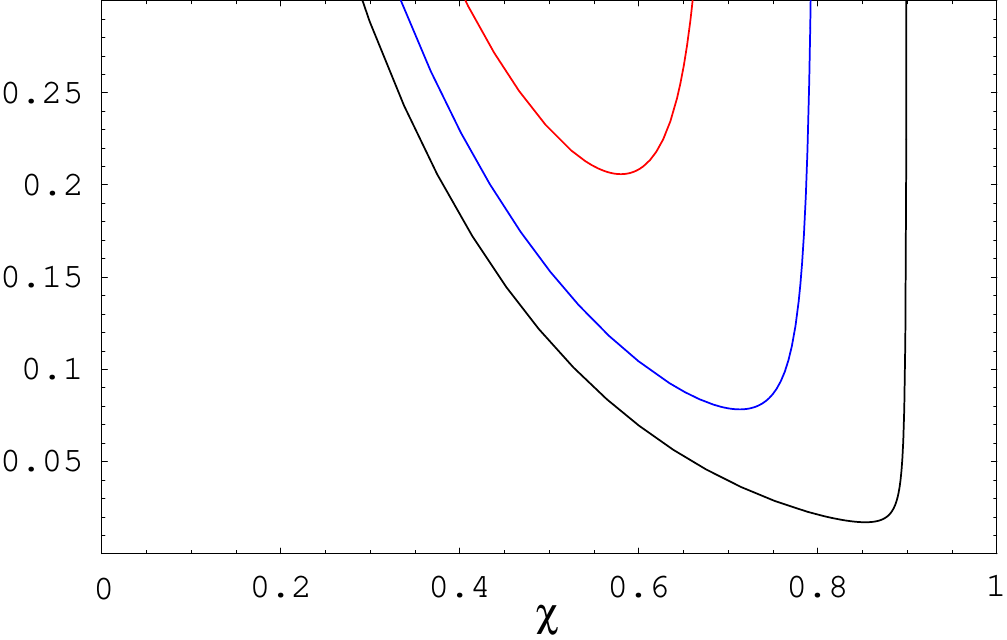}
\includegraphics[width=2.7in]{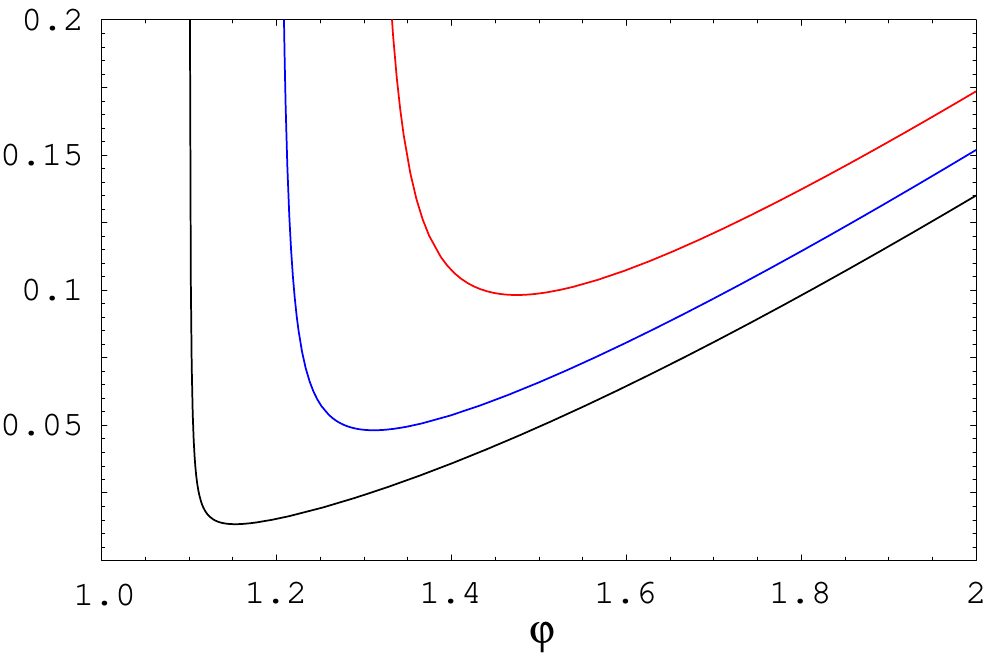}
\caption{
Plots of the functions in (\ref{J1}), (\ref{J2}) in the Black-Scholes model
with $\sigma=1$ for $K/S_0=0.7,0.8,0.9$ (left) and $K/S_0=1.1,1.2,1.3$ (right).
The infimum of the function gives the rate function in the BS model.}
\label{Fig:J}
\end{figure}

%%%%%%%%%%%%%%%%%%%%%%%%%%%%%%%%%%%%%%%%%%%%%%%%%%%%%%%%%%%%%%%%%%%%%%%%%%%%%%%%%%%%%%%
Using the result of Proposition~\ref{prop:1} we can also prove that: 

\begin{proposition}\label{prop:2}
The rate function $\mathcal{I}(K,S_0)$
is a monotonically increasing function of $K$ for $K > S_0$ and monotonically
decreasing function of $K$ for $K  <S_0$. 
\end{proposition}

\subsection{Black-Scholes Model}

In the Black-Scholes model the volatility is constant $\sigma(S)=\sigma$. 
The expression for the rate function $\mathcal{I}(K,S_0)$
simplifies and is given by the following result.

\begin{proposition}\label{RateFunctionBS}
The rate function for small maturity asymptotics for Asian options 
in the Black-Scholes model is 
\begin{equation}\label{IBS}
\mathcal{I}_{\rm BS}(K, S_0) = \frac{1}{\sigma^2} \mathcal{J}_{\rm BS}(K/S_0),
\end{equation}
where $\mathcal{J}_{\rm BS}(K/S_0)$ depends only on the ratio $K/S_0$ and is given by
\begin{equation}\label{JBSresult}
\mathcal{J}_{\rm BS}(K/S_0) 
= 
\begin{cases}
\frac12 \beta^2 - \beta \tanh \left(\frac{\beta}{2}\right) & K \geq S_0 
\\
2\xi (\tan\xi - \xi) & K \leq S_0 
\end{cases}.
\end{equation}
For $K\geq S_0$, $\beta$ is the solution of the equation
\begin{equation}\label{betaeq}
\frac{1}{\beta} \sinh\beta = \frac{K}{S_0},
\end{equation}
and for $K\leq S_0$, $\xi$ is the solution in the interval $[0,\frac{\pi}{2}]$
of the equation 
\begin{equation}\label{lambdaeq}
\frac{1}{2\xi} \sin(2\xi)
 = \frac{K}{S_0}\,.
\end{equation}
\end{proposition}
We note that the two equations (\ref{betaeq}) and (\ref{lambdaeq}) can be written
in a common form as $\frac{1}{z}\sin z = K/S_0$ by denoting $z = 2\xi = 
i\beta$. With this notation, the rate function is 
$\mathcal{J}_{\rm BS}(K/S_0)=z \tan(\frac12 z)-\frac12 z^2$.
We computed the function $\mathcal{J}_{\rm BS}(K/S_0)$ numerically, and the 
plot of this function is shown in Figure~\ref{Fig:Jcal}.

\begin{figure}[t]
\centering
\includegraphics[width=2.5in]{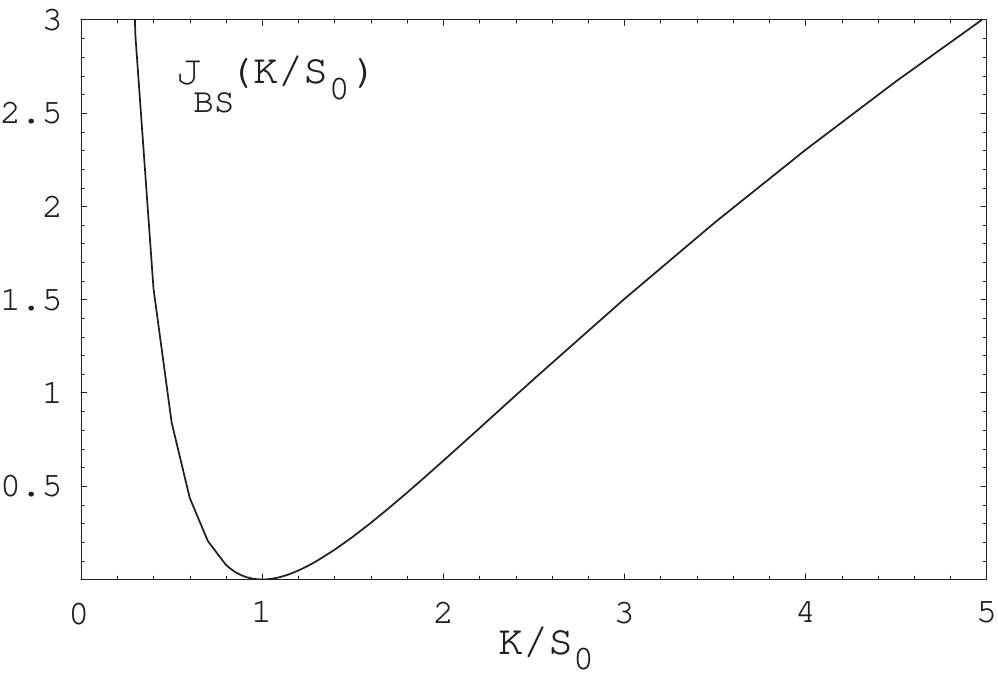}
\includegraphics[width=2.35in]{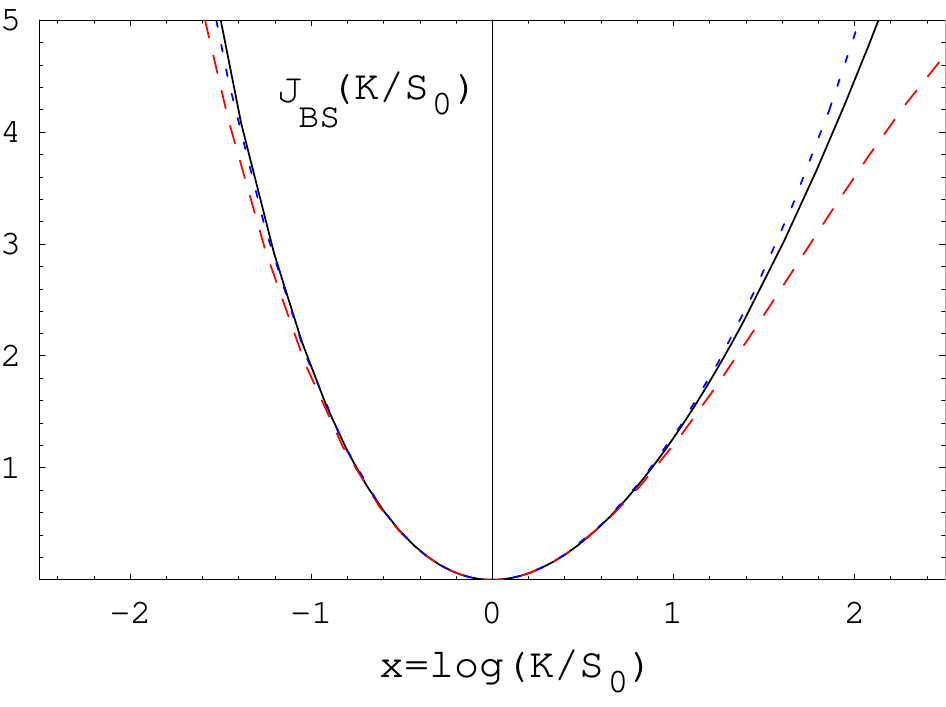}
\caption{
The rate function $\mathcal{J}_{\rm BS}(K/S_0)$ for the Asian options in the
Black-Scholes model.
This is related to $\mathcal{I}_{\rm BS}(K,S_0)$ as in (\ref{IBS}). 
Left: $\mathcal{J}_{\rm BS}(K/S_0)$  vs. $K/S_0$.
Right: the function $\mathcal{J}_{\rm BS}(K/S_0)$ vs. $x=\log(K/S_0)$ (black) 
and the Taylor series obtained by keeping terms up to $O(x^3)$ (dashed red) 
and $O(x^4)$ (dotted blue) in (\ref{TaylorJ}).}
\label{Fig:Jcal}
\end{figure}

We would like to find the series expansion of the Black-Scholes
rate function $\mathcal{J}_{\rm BS}(K/S_0)$ around the ATM point $K/S_0=1$. 
The solutions of the equations for $\beta,\xi$ can be found by inverting 
the series in $k := \frac{K}{S_0}-1$ for (\ref{betaeq}) and 
(\ref{lambdaeq}). 
Substituting into (\ref{JBSresult}) we find the expansion of the rate
function 
\begin{equation}\label{JTaylork}
\mathcal{J}_{\rm BS}\left(\frac{K}{S_0}\right) = \frac32 k^2 - \frac95 k^3 + 
\frac{333}{175} k^4  - \frac{1704}{875} k^5 + O(k^6)\,.
\end{equation}

A similar expansion can be derived in powers of the log-strike 
$x := \log(K/S_0)$
\begin{equation}\label{TaylorJ}
\mathcal{J}_{\rm BS}\left(\frac{K}{S_0}\right) = 
\frac32 x^2 - \frac{3}{10} x^3 + \frac{109}{1400} x^4
- \frac{117}{7000} x^5 + O(x^6)\,.
\end{equation}

Figure~\ref{Fig:Jcal} (right plot) shows the  approximations for the rate 
function $\mathcal{J}_{\rm BS}(K/S_0)$ obtained by keeping the first few terms 
in the series expansion (\ref{TaylorJ}). Keeping the first four terms in 
the expansion (\ref{TaylorJ}) matches the exact result for the rate function 
to an accuracy better than $1.2\%$ for $x = \log(K/S_0) \in (-1.6,1.5)$.

Next we consider the large/small-strike asymptotics of the rate function 
$\mathcal{J}_{\rm BS}(K/S_0)$. This is given by the following result.
\begin{proposition}\label{prop:asympt}
We have
\begin{align}\label{JBSasympt}
& \mathcal{J}_{\rm BS} \left( \frac{K}{S_0} \right) = 
\left\{
\begin{array}{lc}
\frac12 x^2 + x\log (2x) - x + 3\log^2(2x)-2\log(2x) + O(x^{-1})\,,& K\to \infty \\
2e^{-x} -2 - \frac{\pi^2}{2}+ O(e^x)\,,& K \to 0\\
\end{array}
\right.
\end{align}
with $x=\log(K/S_0)$ the log-strike.
\end{proposition}

\subsection{General local volatility model} The rate
function $\mathcal{I}(K,S_0)$ for the general local volatility function
$\sigma(S)$ is given by Proposition~\ref{VarProp}. We give next an 
explicit result for the first three terms in the series expansion of 
$\mathcal{I}(K,S_0)$ in powers of log-strike $x = \log(K/S_0)$.

\begin{proposition}\label{LVexpRF}
The first three terms in the series expansion of the rate function
$\mathcal{I}(K,S_0)$ given by Proposition~\ref{VarProp} in powers of
log-strike $x = \log(K/S_0)$ are 
\begin{eqnarray}\label{RFexp}
\mathcal{I}(K,S_0) &=& \frac{1}{b_1^2} \left\{
\frac32 x^2 + \left( - \frac{3}{10} - \frac{18}{5} \frac{b_2}{b_1^2} 
\right) x^3 \right. \\
& & \left. + \left( \frac{109}{1400} + \frac{117}{175} \frac{b_2}{b_1^2}
+ \frac{1872}{175} \frac{b_2^2}{b_1^4} - \frac{162}{35} \frac{b_3}{b_1^3}\right)
x^4 + O(x^5) \right\}\,.\nonumber
\end{eqnarray}
The coefficients $b_i$ depend on the local volatility function $\sigma(S)$, 
and are given by 
\begin{equation}\label{Ydef}
Y(z) = Z^{-1}(z) = \sum_{i=1}^\infty b_i z^i\,,
\end{equation}
which is the inversion of the power series for the function
\begin{equation}\label{Zdef}
Z(y) = \int_0^y \frac{dw}{\sigma(S_0 e^w)} = \sum_{i=1}^\infty a_i y^i,
\end{equation}
where we assumed that the local volatility function $\sigma(S)$ is sufficiently 
regular such that all the required derivatives exist and are finite.
In particular the expansion (\ref{RFexp}) requires that $\sigma(S)$ is twice 
differentiable.
\end{proposition}

\begin{remark}
For the Black-Scholes model we have $\sigma(S)=\sigma$ which gives
$Z(y)=\frac{1}{\sigma} y$ and thus $Y(z) = \sigma z$. This gives
$b_1=\sigma$ and $b_j=0$ for $j>1$. This recovers the expansion of the
rate function in  the Black-Scholes model given in equation~(\ref{TaylorJ}).
\end{remark}

\begin{corollary}
Assume that $\sigma(\cdot)$ is twice differentiable. 
The expansion (\ref{RFexp}) of the rate function $\mathcal{I}(K,S_0)$ for Asian
options in the local volatility model (\ref{dynamics}) with local volatility 
function $\sigma(S)$ is given in a more explicit form as
\begin{eqnarray}\label{RFexp2}
&& \mathcal{I}(K,S_0) = \frac{1}{\sigma^2(S_0)} \left\{
\frac32 x^2 + \left( - \frac{3}{10} - \frac{9}{5} S_0 \frac{\sigma'(S_0)}{\sigma(S_0)} 
\right) x^3  \right. \\
&& \left. + \left( \frac{109}{1400} - \frac{153}{350} S_0 \frac{\sigma'(S_0)}{\sigma(S_0)} +
\frac{333}{175} S_0^2 \left( \frac{\sigma'(S_0)}{\sigma(S_0)} \right)^2 -
\frac{27}{35} S_0^2 \frac{\sigma''(S_0)}{\sigma(S_0)} \right) x^4 + O(x^5) \right\}\,. \nonumber
\end{eqnarray}
\end{corollary}

This is obtained by noting that the coefficients $a_i$ can be obtained by taking 
derivatives with respect to $y$ of (\ref{Zdef}) at $y=0$. We get
\begin{eqnarray}
a_1 &=& \frac{1}{\sigma(S_0)}\, ,\\
a_2 &=& - \frac12 S_0 \frac{\sigma'(S_0)}{\sigma^2(S_0)}\, ,\\
a_3 &=& - \frac16 S_0 \frac{\sigma'(S_0)}{\sigma^2(S_0)} + \frac13 S_0^2 
\frac{[\sigma'(S_0)]^2}{\sigma^3(S_0)} - 
\frac16 S_0^2 \frac{\sigma''(S_0)}{\sigma^2(S_0)}\, .
\end{eqnarray}
Inverting the Taylor series (\ref{Zdef}) we find the coefficients $b_i$:
\begin{eqnarray}
&& b_1 = \frac{1}{a_1}  = \sigma(S_0) \, ,\\
&& b_2 = - \frac{a_2}{a_1^3} = \frac12 S_0 \sigma'(S_0) \sigma(S_0) \, ,\\
&& b_3 = \frac{2a_2^2}{a_1^5}- \frac{a_3}{a_1^4}  = 
\frac16 S_0^2 [\sigma'(S_0)]^2 \sigma(S_0) +
\frac16 S_0    \sigma'(S_0) \sigma^2(S_0) +
\frac16 S_0^2  \sigma''(S_0) \sigma^2(S_0) \,.
\end{eqnarray}
Substituting into (\ref{RFexp}) we obtain the result (\ref{RFexp2})
for the rate function of the Asian option in the LV model to order $O(x^3)$,
expressed only in terms of the ATM local volatility and its derivatives.

%%%%%%%%%%%%%%%%%%%%%%%%%%%%%%%%%%%%%%%%%%%%%%%%%%%%%%%%%%%%%%%%%%%%%%
%%%%%%%%%%%%  Implied volatility %%%%%%%%%%%%%%%%%%%%%%%%%%%%%%%%%%%%%
%%%%%%%%%%%%%%%%%%%%%%%%%%%%%%%%%%%%%%%%%%%%%%%%%%%%%%%%%%%%%%%%%%%%%%

\section{Implied Volatility and Numerical Tests}
\label{sec:4}

Implied volatility for European options was extensively studied in the 
mathematical finance literature. Due to the lack of a simple closed-form 
formula as in the Black-Scholes model for European options, the implied 
volatility for Asian options was much less studied. Our work on the short 
maturity asymptotics for the price of the Asian options for local volatility 
models can shed some light on the short maturity implied volatility for the 
Asian options.

\subsection{Implied Volatility}

The implied volatility $\sigma_{\text{implied}}$ is defined as the constant
volatility which must be used in the Black-Scholes model for the Asian option
such that its price matches that of the Asian option in the local volatility 
model. That is, we must have
\begin{equation}
C_{BS}(K,S_{0},\sigma_{\text{implied}},T)=C(K,S_{0},T),
\end{equation}
where $C(K,S_{0},T)=e^{-rT}\mathbb{E}[(\frac{1}{T}\int_{0}^{T}S_{t}dt-K)^{+}]$, 
where $S_{t}$ satisfies the dynamics \eqref{dynamics} with the local volatility 
$\sigma(\cdot)$ and $C_{BS}(K,S_{0},\sigma_{\text{implied}},T)=
e^{-rt}\mathbb{E}[(\frac{1}{T}\int_{0}^{T}S_{t}dt-K)^{+}]$, 
where $S_{t}$ satisfies the dynamics \eqref{dynamics} with $\sigma(\cdot)
\equiv\sigma_{\text{implied}}$.

The fact that the implied volatility for an Asian option is
well defined is not a trivial fact. It follows from the fact that the 
Vega of an Asian option in the Black-Scholes model is always positive, 
which was proved in Carr et al. \cite{CarrVega}, and hence, if one considers 
its price as a function of the volatility, the inverse of this function exists. 
As the volatility parameter $\sigma$ is increased from zero to 
infinity, the price of an out-of-the-money Asian option goes from
$C_{BS}(K,S_{0},0,T)=e^{-rT}\left(\frac{1}{T}\int_{0}^{T}S_{0}e^{(r-q)t}dt-K\right)^{+}$
to $C_{BS}(K,S_{0},\infty,T)=e^{-rT}\frac{1}{T}\int_{0}^{T}S_{0}e^{(r-q)t}dt$. 
Hence, the implied volatility $\sigma_{\text{implied}}$ is well defined.
Note that it is trivial to see $C_{BS}(K,S_{0},0,T)=e^{-rT}\left(\frac{1}{T}\int_{0}^{T}S_{0}e^{(r-q)t}dt-K\right)^{+}$
but the statement $C_{BS}(K,S_{0},\infty,T)=e^{-rT}\frac{1}{T}\int_{0}^{T}S_{0}e^{(r-q)t}dt$
is less trivial and we will give a rigorous proof of this statement in Proposition \ref{InfiniteSigma} in the Appendix. 

>From Theorem \ref{MainThm} for out-of-the-money Asian options, 
we can obtain the short-maturity limit of the 
implied volatility 
$\sigma_{\text{implied}}$. For simplicity, we only
give a proof for the case $r=q=0$. We expect
the same result holds for general $r,q$.

\begin{proposition}\label{prop:18}
Assume $r=q=0$ and \eqref{assumpI} and \eqref{assumpII} hold.

(i) The $T\to 0$ limit of the implied volatility of out-of-the-money Asian 
options in the local volatility model (\ref{dynamics}) is given by 
\begin{equation}\label{sigmaimp}
\lim_{T\to 0} \sigma_{\text{implied}}^{2}(K,S_{0},T)
= \frac{\mathcal{J}_{\rm BS}(K/S_0)}{\mathcal{I}(K,S_0)}\,,
\end{equation}
where $\mathcal{I}(K,S_0)$ is given in Proposition~\ref{VarProp} and
$\mathcal{J}_{\rm BS}(K/S_0)$ is given in Proposition~\ref{RateFunctionBS}.

(ii) Under the same assumptions on $\sigma(\cdot)$ as in Theorem~\ref{ATMThm},
the $T\to 0$ limit of the implied volatility of at-the-money Asian options 
in the local volatility model (\ref{dynamics}) is given by 
\begin{equation}
\lim_{T\rightarrow 0}\sigma_{\text{implied}}(K,S_{0},T)
=\sigma(S_{0}).
\end{equation}
\end{proposition}

\subsection{Equivalent Black-Scholes and Bachelier Volatility of Asian Options}

One can define an equivalent Black-Scholes volatility of an Asian option,
as that value of the volatility for which the Black-Scholes price of an
European (vanilla) option with maturity $T$ and underlying value $A(T)$ reproduces the price of the Asian
option with the same maturity $T$. We will denote this volatility as 
$\Sigma_{\rm LN}(K,S_0,T)$. We have thus
\begin{eqnarray}\label{Casympt}
C(K,S_0,T) &=& e^{-rT} [A(T) \Phi(d_1) - K \Phi(d_2)] \\
P(K,S_0,T) &=& e^{-rT} [K \Phi(-d_2) - A(T) \Phi(-d_1)] \nonumber
\end{eqnarray}
where $A(T)$ is given in (\ref{Adef})
and $d_{1,2} = \frac{1}{\Sigma_{LN}\sqrt{T}} ( \log(A(T)/K) \pm \frac12\Sigma^2_{LN}T)$.
We remark that this is a natural definition since
$A(T)$ is the forward price of the underlying of the Asian option $\mathbb{E}[\frac{1}{T}\int_{0}^{T}S_{t}dt]$.
Using (\ref{Casympt}) ensures that put-call parity for Asian options 
(\ref{PCparity}) holds; this would not hold
for example if one used $S_{0}$ instead of $A(T)$. 
The relation (\ref{Casympt}) also reproduces correctly the price of an
Asian call option with zero strike $K=0$: $C(0,S_{0},T) = e^{-rT} 
\mathbb{E}[\frac{1}{T} \int_0^T S_t dt] = e^{-rT} A(T)$.
A similar equivalent normal volatility $\Sigma_{\rm N}(K,S_0,T)$ of an Asian 
option can be defined in terms of the Bachelier option pricing formula.

The equivalent log-normal volatility $\Sigma_{LN}$ defined as in (\ref{Casympt})
exists for any Asian call option price $C(K,S_0,T)$ satisfying the 
bounds $(A(T) - K)^+ \leq e^{rT} C(K,S_0,T) \leq A(T)$ \cite{RR}. 
These bounds are indeed satisfied for Asian options under the local volatility 
model (\ref{dynamics}). The lower bound is satisfied by the convexity of the 
payoff $(x-K)^+$, and the upper bound follows from $(x-K)^+ \leq x$.

Although Asian options are quoted in practice by price, and not implied
volatility, such equivalent volatilities are a convenient representation 
of the short maturity asymptotics of the Asian option prices. Also, they give
a natural formulation for the small-maturity asymptotics of the Asian options
in Proposition~\ref{VarProp}, which is equivalent to a small-maturity limit 
for these volatilities.
This result is somewhat similar to the representation of the small maturity 
asymptotics for European options in the local volatility model in terms of 
the BBF formula for the implied volatility \cite{BBF}.

The short maturity asymptotics for out-of-the-money Asian options given in
Theorem~\ref{MainThm} gives the following short-time asymptotics 
for the equivalent volatilities of the Asian options in the local volatility
model (\ref{dynamics}). For simplicity, we only prove the $r=q=0$ case.
We expect the same result holds for general $r,q$.

\begin{proposition}\label{propSigBS}
Assume $r=q=0$ and \eqref{assumpI} and \eqref{assumpII} hold.

(i) The short-time limit $T\to 0$ of the Black-Scholes equivalent volatility
of an out-of-the-money Asian option is given by
\begin{equation}\label{sigBSLV}
\lim_{T\to 0} \Sigma_{\rm LN}^2(K,S_0,T) = \frac12 
\frac{\log^2\left(\frac{K}{S_0}\right)}{\mathcal{I}(K,S_0)},
\end{equation}
and the corresponding result for the Bachelier equivalent volatility is
\begin{equation}\label{sigNLV}
\lim_{T\to 0} \Sigma_{\rm N}^2(K,S_0,T) = \frac12 
\frac{\left(\frac{K}{S_0}-1\right)^2}{\mathcal{I}(K,S_0)},
\end{equation}
where $\mathcal{I}(K,S_0)$ is given in Proposition~\ref{VarProp}.

(ii) Under assumptions of Theorem \ref{ATMThm} on $\sigma(\cdot)$, 
the short-time limit $T\to 0$ of the Black-Scholes equivalent volatility
of an at-the-money Asian option is given by
\begin{equation}
\lim_{T\to 0} \Sigma_{\rm LN}(K,S_0,T) = \frac{1}{\sqrt{3}}\sigma(S_{0}),
\end{equation}
and the corresponding result for the Bachelier equivalent volatility is
\begin{equation}
\lim_{T\to 0} \Sigma_{\rm N}(K,S_0,T) = \frac{1}{\sqrt{3}}\sigma(S_{0})S_{0}.
\end{equation}
\end{proposition}

In the Black-Scholes model we can get explicit results for the equivalent
volatilities, using the result for the rate function 
$\mathcal{J}_{\rm BS}(K/S_0)$ derived in Proposition~\ref{RateFunctionBS}. 
Denote the corresponding equivalent volatilities $\Sigma_{\rm LN}^{(BS)}$
and $\Sigma_{\rm N}^{(BS)}$. Their expansions about the ATM point in
powers of log-strike $x=\log(K/S_0)$ are 
\begin{equation}\label{SigBSexp}
\Sigma_{\rm LN}^{(BS)}(K/S_0) = \frac{1}{\sqrt3} \sigma \left( 1 + \frac{1}{10} x -
\frac{23}{2100} x^2 + \frac{1}{3500} x^3 
+ O(x^4) \right),
\end{equation}
and in powers of $k=K/S_0-1$, respectively
\begin{equation}
\Sigma_{\rm N}^{(BS)}(K/S_0) = \frac{1}{\sqrt3} \sigma S_0  
\left( 1 + \frac35 k - \frac{33}{350} k^2 + \frac{83}{1750} k^3 + O(k^4)
\right)\,.
\end{equation}

Using the equation (\ref{RFexp2}) we can derive similar results for a more
general local volatility function $\sigma(S)$. From the first three
terms in the Taylor expansion of the rate function in powers of log-strike
we can get explicit results for the ATM equivalent volatility, the skew and
smile convexity at the ATM point. We use $\Sigma_{\rm LN(N)}(K,S_{0})$ to denote 
$\lim_{T\rightarrow 0}\Sigma_{\rm LN(N)}(K,S_{0},T)$.

\begin{proposition}\label{prop:SigBSLV}
Assume $\sigma(\cdot)$ is twice differentiable. 
The small maturity limit $T\to 0$ of the log-normal equivalent volatility
for an Asian option in the local volatility model (\ref{dynamics}) has the
expansion in $x=\log(K/S_0)$ around the ATM point
\begin{eqnarray}
&& \Sigma_{\rm LN}(K,S_0) = \frac{1}{\sqrt3} \sigma(S_0)\left\{
1 + \left( \frac{1}{10} + \frac35 S_0 \frac{\sigma'(S_0)}{\sigma(S_0)} \right)x 
\right. \\
&& + \left.
\left(-\frac{23}{2100} + \frac{57}{175} S_0 \frac{\sigma'(S_0)}{\sigma(S_0)}
- \frac{33}{350} S_0^2 \left(\frac{\sigma'(S_0)}{\sigma(S_0)}\right)^2
+ \frac{9}{35} S_0^2 \frac{\sigma''(S_0)}{\sigma(S_0)} \right) x^2 + O(x^3) \right\}\, .
\nonumber
\end{eqnarray}
The corresponding expansion for the normal equivalent volatility of an Asian
option has the expansion in $k=\frac{K}{S_0}-1$ around the ATM point
\begin{eqnarray}
&& \Sigma_{\rm N}(K,S_0) = \frac{1}{\sqrt3} S_0 \sigma(S_0)\left\{
1 + \left( \frac{3}{5} + \frac35 S_0 \frac{\sigma'(S_0)}{\sigma(S_0)} \right)k 
\right. \\
&& + \left.
\left(-\frac{33}{350} + \frac{57}{175} S_0 \frac{\sigma'(S_0)}{\sigma(S_0)}
- \frac{33}{350} S_0^2 \left(\frac{\sigma'(S_0)}{\sigma(S_0)}\right)^2
+ \frac{9}{35} S_0^2 \frac{\sigma''(S_0)}{\sigma(S_0)} \right) k^2 + O(k^3) \right\}\, .
\nonumber
\end{eqnarray}
\end{proposition}

The result of Proposition~\ref{prop:SigBSLV}
is similar to the expansion of the implied volatility obtained
in an uncorrelated local-stochastic volatility model \cite{FordeJacquier}
for the implied volatility of vanilla European options, see Theorem 4.1
in \cite{FordeJacquier}, giving  the level, slope and convexity of the small-time
implied volatility. Indeed, this result allows pricing Asian options directly
from the observable European volatility skew, assuming local volatility
dynamics, but without assuming a particular form of the local volatility 
function!

Combining the result of Proposition~\ref{prop:SigBSLV} and the well-known
result for the ATM skew of European options in the local volatility model,
see e.g. \cite{FordeJacquier}, we have:
\begin{remark}
In the short maturity limit,
\begin{align}
&\text{ATM Asian vol}=\frac{1}{\sqrt{3}}\cdot\text{ATM European vol},
\\
&\text{ATM Asian skew slope}=\frac{1}{\sqrt{3}}\cdot
\left[\frac{1}{10}\cdot\text{ATM European vol}
+\frac{6}{5}\cdot\text{ATM European skew}\right]\,,
\end{align}
where the slope is with respect to the log-moneyness.
\end{remark}

\subsection{Numerical Tests}

We present in this section a few numerical tests of the short-maturity
asymptotic results for Asian options obtained in this paper. 

For ATM Asian options we have the result of Theorem~\ref{ATMThm} which 
can be used directly to obtain a price.
For out-of-the-money Asian options, we have the asymptotic results of 
Theorem~\ref{MainThm}. In practice, we find that it is convenient to use 
this result to obtain the short-maturity equivalent log-normal volatility
of an Asian option $\Sigma_{\rm LN}(K,S_0)$ 
(or normal volatility $\Sigma_{\rm N}(K,S_0)$), as given in
Proposition~\ref{propSigBS}.

The equivalent log-normal volatility $\Sigma_{\rm LN}$ given by 
Proposition~\ref{propSigBS} can be used in the Black-Scholes pricing formula 
as in (\ref{Casympt}) to compute Asian option prices.
Although this approach introduces subleading terms in the option price which 
are not constrained by the asymptotic result of Theorem~\ref{MainThm}, we will 
demonstrate that it gives predictions that are in good agreement with the 
numerical simulations of the Asian option prices. 

%%%%%%%%%   TABLE 1: OTM Call Asian  %%%%%%%%%%%%%%%%%%%%%%

\begin{table}[t!]
\caption{\label{Table:1} 
Numerical results for call (upper) and put (lower) Asian options with 
maturity $T=0.5, 1, 2$ years and parameters (\ref{scenario1}) obtained by 
MC simulation in the BS model (1 stdev in brackets), compared against the
asymptotic results $C_{\rm as}(K), P_{\rm as}(K)$ given by (\ref{Casympt}). 
The last column gives the asymptotic equivalent log-normal volatility of the 
Asian option $\Sigma_{\rm LN}^{(BS)}$.}
\begin{center}
\begin{tabular}{|c|cc|cc|cc|c|}
\hline
  & \multicolumn{2}{|c|}{$T=0.5$}  &   \multicolumn{2}{|c|}{$T=1$}  & 
    \multicolumn{2}{|c|}{$T=2$}  &  \\
%  &          &                &                     &        &  \\
\hline
$K$ & MC(n=800) & $C_{\rm as}(K)$ &  MC(n=800)  & $C_{\rm as}(K)$  & MC(n=800) & 
    $C_{\rm as}(K)$  & $\Sigma_{\rm LN}^{(BS)}$ \\
\hline
\hline
100 & $4.8871(0.0078)$   & 4.8830  & $6.9037(0.0115)$  & 6.9013 & 
      $9.7417(0.0155)$  & 9.7477  & 17.32\%  \\
105 & $2.9205(0.0062)$   & 2.9188  & $4.8848(0.0098)$  & 4.8847 & 
      $7.7268(0.0155)$  & 7.7382   & 17.41\%  \\
110 & $1.6372(0.0046)$   & 1.6388   & $3.3689(0.0083)$  & 3.3715 & 
      $6.0737(0.0139)$  & 6.0826  & 17.48\%  \\
115 & $0.8650(0.0033)$   & 0.8671   & $2.2698(0.0068)$  & 2.2745 & 
      $4.7370(0.0124)$  & 4.7505  & 17.56\% \\
120 & $0.4336(0.0023)$   & 0.4351   & $1.4980(0.0055)$  & 1.5033 & 
      $3.6692(0.0110)$  & 3.6835  & 17.63\% \\
125 & $0.2075(0.0016)$   & 0.2081 & $0.9715(0.0045)$    & 0.9758 & 
      $2.8254(0.0098)$  & 2.8414 & 17.70\%  \\
130&  $0.0949(0.0011)$   & 0.0953  & $0.6201(0.0036)$    & 0.6234 & 
      $2.1657(0.0086)$  & 2.1790 & 17.76\% \\
\hline
\hline
$K$ & MC(n=800) & $P_{\rm as}(K)$ &  MC(n=800)  & $P_{\rm as}(K)$  & MC(n=800) & 
    $P_{\rm as}(K)$  & $\Sigma_{\rm LN}^{(BS)}$ \\
\hline
\hline
70 & $0.0034(0.0001)$   & 0.0035  & $0.0810(0.0007)$  & 0.0809 & 
     $0.5580(0.0024)$  & 0.5596  & 16.68\%  \\
75 & $0.0264(0.0004)$   & 0.0263  & $0.2579(0.0014)$  & 0.2580 & 
     $1.1220(0.0036)$  & 1.1250 & 16.81\%  \\
80 & $0.1296(0.0009)$   & 0.1295   & $0.6608(0.0025)$  & 0.6609 & 
     $2.0100(0.0050)$  & 2.0167 & 16.92\%  \\
85 & $0.4548(0.0018)$   & 0.4543   & $1.4221(0.0038)$  & 1.4237 & 
     $3.2880(0.0067)$  & 3.2984 & 17.03\%  \\
90 & $1.2187(0.0031)$   & 1.2190   & $2.6671(0.0054)$  & 2.6711 & 
     $4.9963(0.0085)$  & 5.0095 & 17.14\%  \\
95 & $2.6475(0.0048)$   & 2.6494 & $4.4820(0.0072)$  & 4.4877 & 
     $7.1464(0.0103)$  & 7.1628 & 17.23\%  \\
100& $4.8789(0.0066)$   & 4.8830 & $6.8928(0.0090)$ & 6.9013 & 
     $9.7280(0.0121)$  & 9.7477 & 17.32\% \\
\hline
\end{tabular}
\end{center}
\end{table}

%%%%%%%%%%%%%%%%%%%%%%%%%%%%%%%%%%%%%%%%%%%%%%%%%%%%%%%%%%%%%%%%%%

We consider next a few numerical tests of the asymptotic pricing formulas,
on the example of the Asian options in the Black-Scholes model. 
One first test assumes the model parameters
\begin{equation}\label{scenario1}
r = q = 0, S_0 = 100, \sigma = 30\%.
\end{equation}
In Table \ref{Table:1} we show the prices of out-of-the-money 
call and put Asian options with maturities $T = 0.5, 1, 2$ years, 
comparing the results of a Monte Carlo calculation against the asymptotic 
results $C_{\rm as}$ and $P_{\rm as}$ given in (\ref{Casympt}).
The Monte Carlo simulation was performed with $N=10^6$ paths, and the time
line was discretized with $n=800$ time steps. 
We note very good agreement of the short-maturity asymptotic results
with the results of the Monte Carlo calculation. The asymptotic results 
are always within one standard deviation (68\% CL) of the MC result for $T= 0.5, 1$,
and within two standard deviations (95\% CL) for $T=2$.

%%%%%%%%%%%%%%%%%%%  7 scenarios %%%%%%%%%%%%%%%%%%%%%%%%%%%%%%%%%%

We also compare with the benchmark scenarios proposed in \cite{FMW} and
which were commonly used in the literature on pricing Asian options
\cite{Dassios,DufresneLaguerre,Linetsky,FPP2013,VecerXu}.
We show in Table~\ref{Table:2} numerical results for the asymptotic 
approximation for the Asian options obtained from (\ref{Casympt}), for the
scenarios proposed in \cite{FMW}.
They are compared against the very precise results of \cite{Linetsky} obtained 
using a spectral expansion, and against the simple Levy approximation 
\cite{Levy}. The asymptotic result performs in all cases better than the Levy
approximation, and the agreement with the results of \cite{Linetsky} is 
better than 0.7\% in all cases. 

\begin{table}[t!]
\caption{\label{Table:2} 
Numerical results for Asian call options in the Black-Scholes model
under the benchmark scenarios considered in \cite{FMW,Linetsky}. 
The last 4 columns show: a) the results from
the asymptotic expansion (\ref{Casympt}) of this paper (PZ), b) the
3rd order approximation proposed in Foschi et al. \cite{FPP2013} (FPP3),
c) Levy 
approximation \cite{Levy}, d) precise evaluation using the spectral 
expansion in \cite{Linetsky}.}
\begin{center}
\begin{tabular}{ccccc|c|c|c|c}
\hline
$r$ & $T$ & $S_0$ & $K$ & $\sigma$ 
         & PZ & FPP3 & Levy & Linetsky \\
\hline
\hline
0.02   & 1 & 2    & 2 & 0.1  & 0.055923 & 0.055986 & 0.056054 & 0.055986 \\
0.18   & 1 & 2    & 2 & 0.3  & 0.217054 & 0.218387 & 0.219829 & 0.218387 \\
0.0125 & 2 & 2    & 2 & 0.25 & 0.172163 & 0.172267 & 0.173490 & 0.172269 \\
0.05   & 1 & 1.9  & 2 & 0.5  & 0.192895 & 0.193164 & 0.195379 & 0.193174 \\
0.05   & 1 & 2    & 2 & 0.5  & 0.246125 & 0.246406 & 0.249791 & 0.246416 \\
0.05   & 1 & 2.1  & 2 & 0.5  & 0.305927 & 0.306210 & 0.310646 & 0.306220 \\
0.05   & 2 & 2    & 2 & 0.5  & 0.349314 & 0.350040 & 0.359204 & 0.350095 \\
\hline
\end{tabular}
\end{center}
\end{table}

%%%%%%%%%%%%%%%%%%%%%%%%%%%%%%%%%%%%%%%%%%%%%%%%%%%%%%%%%%%%%%%%%%%%

The numerical tests presented show that the short-maturity asymptotic 
results of this paper can be used as a good approximation for 
Asian options prices with maturities relevant for practical applications.

%%%%%%%%%%%%%%%%%%%%%%%%%%%%%%%%%%%%%%%%%%%%%%%%%%%%%%%%%%%%%%%%%%%%%

\section{Asymptotics for Floating Strike Asian Options}
\label{sec:floating}

There are many variations of the standard Asian options in the finance literature
and one of the most used is the so-called floating strike Asian options.
The price of the floating strike Asian call/put options are given by
\begin{align}
&C_f(T):=e^{-rT}\mathbb{E}\left[\left(\kappa S_{T}-\frac{1}{T}\int_{0}^{T}S_{t}dt\right)^{+}\right],
\\
&P_f(T):=e^{-rT}\mathbb{E}\left[\left(\frac{1}{T}\int_{0}^{T}S_{t}dt-\kappa S_{T}\right)^{+}\right],
\end{align}
where $\kappa>0$ is the strike, see e.g.
\cite{Levy,Ritchken,Alziary,Chung,RogersShi,HW}.
The floating-strike Asian option is more difficult to price than the fixed-strike case
because the joint law of $S_{T}$ and $\frac{1}{T}\int_{0}^{T}S_{t}dt$ is needed, and also
the one dimensional PDE that the floating-strike Asian price satisfies after a change of num\'{e}raire
is difficult to solve numerically as the Dirac delta function appears as a coefficient, 
see e.g. \cite{RogersShi,Alziary}.

When $\kappa<1$, the call option is OTM, the put option is ITM;
when $\kappa>1$, the call option is ITM, the put option is OTM;
when $\kappa=1$, the call/put options are ATM. We are interested
in the short maturity, i.e. $T\rightarrow 0$ asymptotics of these options.

For the Black-Scholes model, it was shown by Henderson and 
Wojakowski \cite{HW} that the floating-strike Asian options with continuous 
time averaging can be related to fixed strike ones as
\begin{align}\label{HWsymm}
& e^{-rT} \mathbb{E}[(\kappa S_T - A_T)^+]  = 
e^{-q T} \mathbb{E}_*[(\kappa S_0 - A_T)^+]\, , \\
& e^{-rT} \mathbb{E}[(A_T - \kappa S_T)^+]  = 
e^{-q T} \mathbb{E}_*[(A_T - \kappa S_0)^+] \,. \nonumber
\end{align}
The expectations on the right-hand side are taken with respect to a different 
measure $\mathbb{Q}_*$, where the asset price $S_t$ is given by the process
\begin{equation}
dS_t = (q-r) S_t dt + \sigma S_t dW_t^* \, ,
\end{equation}
with $W_t^*$ a standard Brownian motion in the $\mathbb{Q}_*$ measure.

However, in our general setting of local volatility models, the equivalence
relations (\ref{HWsymm}) do not hold, and hence the asymptotics for floating 
strike Asian options must be obtained independently from that of the fixed 
strike Asian options. 
But it is not difficult to observe that the same techniques, i.e. large 
deviations, calculus of variation, Gaussian approximations, which were used
to obtain the short-maturity asymptotics for the fixed strike 
Asian options, can be applied with little modification for the floating strike 
ones. For the sake of simplicity, we will only provide a sketch of the proofs 
in the Appendix.

\begin{proposition}\label{FloatingThm}
(i) Under the same assumptions as in Theorem \ref{MainThm}, when $\kappa<1$, 
\begin{align}
&C_f(T)=e^{-\frac{1}{T}\mathcal{I}_{f}(\kappa,S_{0})+o(\frac{1}{T})},
\\
&P_f(T)=(1-\kappa)S_{0}-\frac{S_{0}}{2}(r+q)T+\kappa S_{0}qT+O(T^{2}),
\end{align}
as $T\rightarrow 0$, where
\begin{equation}\label{IfEqn}
\mathcal{I}_{f}(\kappa,S_{0})=
\inf_{\substack{
\int_{0}^{1}e^{g(t)}dt=\kappa e^{g(1)}
\\
g(0)=\log S_{0}, g\in\mathcal{AC}[0,1]}}
\frac{1}{2}\int_{0}^{1}\left(\frac{g'(t)}{\sigma(e^{g(t)})}\right)^{2}dt.
\end{equation}

(ii) Under the same assumptions as in Theorem \ref{MainThm}, when $\kappa>1$, 
\begin{align}
&P_f(T)=e^{-\frac{1}{T}\mathcal{I}_{f}(\kappa,S_{0})+o(\frac{1}{T})},
\\
&C_f(T)=(\kappa-1)S_{0}+\frac{S_{0}}{2}(r+q)T-\kappa S_{0}qT+O(T^{2}),
\end{align}
as $T\rightarrow 0$, where $\mathcal{I}_{f}$ is defined in \eqref{IfEqn}.

(iii) Under the same assumptions as in Theorem \ref{ATMThm}, when $\kappa=1$, 
\begin{equation}
\lim_{T\rightarrow 0}\frac{1}{\sqrt{T}}C_f(T)
=\lim_{T\rightarrow 0}\frac{1}{\sqrt{T}}P_f(T)
=\frac{1}{\sqrt{6\pi}}\sigma(S_{0})S_{0}.
\end{equation}
\end{proposition}

For a general local volatility function $\sigma(S)$, the rate function 
$I_f(\kappa,S_0)$ is given by the following result:

\begin{proposition}\label{VarFloating}
The solution of the variational problem in Proposition~\ref{FloatingThm} 
for the short maturity asymptotics of floating strike Asian options is given by
\begin{equation}\label{Ifres}
\mathcal{I}_f(\kappa, S_0) = \lambda(\kappa-1)e^{f_1} 
  + \frac12\lambda^2 \kappa^2 e^{2f_1} \sigma^2(S_0 e^{f_1}),
\end{equation}
where $f_1=f(1)$ with $f(t)$ given by the solution of the differential equation
\begin{equation}\label{ELf}
\frac{d}{dt} \left( \frac{f'(t)}{\sigma(S_0 e^{f(t)})} \right) = 
\lambda e^{f(t)} \sigma(S_0 e^{f(t)}),
\end{equation}
with
\begin{equation}\label{lambdarel}
\lambda = 2 \frac{1 - e^{f(1)}}{I_s[f]\{I_s[f] - 2\kappa e^{f_1} 
\sigma(S_0 e^{f_1})\} }\,,\qquad
I_s[f] = \int_0^1 ds e^{f(s)} \sigma(S_0 e^{f(s)}) ,
\end{equation}
and boundary conditions
\begin{equation}\label{BCf}
f(0)=0\,,\qquad f'(1) = \lambda \kappa e^{f_1} \sigma^2(S_0 e^{f_1})\,.
\end{equation}
\end{proposition}

We briefly outline in the following a possible numerical method for solving this calculus
problem and finding the rate function $I_f(\kappa,S_0)$ for a given local volatility function 
$\sigma(S)$. This can be reduced to solving a non-linear equation for the variable $\lambda$. 
For a given value
of $\lambda$, one can solve numerically the Euler-Lagrange equation (\ref{ELf}) with the 
boundary conditions (\ref{BCf}), and find the (non-optimal) function $f(t)$. Using this 
function we can
compute the integral $I_s[f]$, and evaluate the expression for $\lambda$ in (\ref{lambdarel}).
Requiring that the result for $\lambda$ following from (\ref{lambdarel}) is the same as the 
input value determines the value of this variable, and thus the optimal function $f(t)$. 
This equation for $\lambda$ can be solved by scanning over $\lambda$.

For the limiting case of the Black-Scholes model $\sigma(S)=\sigma$, the rate
function $I_f(\kappa,S_0)$ can be found exactly, and is related to the 
Black-Scholes rate function $\mathcal{J}_{\rm BS}(K/S_0)$ given in 
Proposition~\ref{RateFunctionBS}. 

\begin{proposition}\label{BSfloating}
In the Black-Scholes model, the rate function of the floating strike Asian
option is given by
\begin{equation}\label{IfBS}
\mathcal{I}_f(\kappa) = 
\inf_{\substack{
\int_{0}^{1}e^{h(t)}dt=\kappa 
\\
h(0)=0, h\in\mathcal{AC}[0,1]}}
\frac{1}{2\sigma^2}\int_{0}^{1} (h'(t))^{2} dt = 
\frac{1}{\sigma^2} \mathcal{J}_{\rm BS}(\kappa)\,,
\end{equation}
where $\mathcal{J}_{\rm BS}(\kappa)$ is the rate function for fixed strike Asian 
options in the Black-Scholes model given in (\ref{JBSresult}).
\end{proposition}

\begin{remark}
Note that for the Black-Scholes model, $\mathcal{I}_{f}(\kappa,S_{0})$ is independent of $S_{0}$
and hence we use the notation $\mathcal{I}_{f}(\kappa)$ in \eqref{IfBS}.
\end{remark}

\begin{remark}
The relation (\ref{IfBS}) is consistent with the equivalence relations 
(\ref{HWsymm}). As noted above, the short maturity asymptotics of the
Asian options is independent of the interest rate $r$ and dividend yield $q$
such that in this limit the expectations in the relation (\ref{HWsymm}) are 
taken in the same measure. This implies that in the short maturity limit,
fixed and floating strike Asian options are equal to each other, up to the
substitution $K/S_0 \mapsto \kappa$.
\end{remark}

It is instructive to compare this explicit solution in  the Black-Scholes model
with the result of Proposition~\ref{VarFloating}. For the case of constant 
volatility function $\sigma(S)=\sigma$, the relation (\ref{lambdarel}) 
simplifies as
\begin{equation}
\lambda_{\rm BS}(\kappa) = \frac{2}{\sigma^2 \kappa^2} e^{-2f(1)}(e^{f(1)} - 1)
 = \frac{2}{\sigma^2 \kappa^2} e^{2h(1)}(e^{-h(1)} - 1) \,.
\end{equation}
where $h(t) = f(1-t) - f(1)$ is the solution of the equivalent variational problem 
in (\ref{IfBS}). This agrees with the solution for the Lagrange multiplier 
$\lambda$ for the auxiliary variational problem (\ref{IfBS}), which is given in 
(\ref{lambdasol}). The value $\lambda_{\rm BS}(\kappa)$ (with $\sigma = \sigma(S_0)$)
can be used as a starting point for scanning over the values of $\lambda$ in the numerical 
solution  of the variational problem for the general local volatility model.

%%%%%%%%%%%%%%%%%%%%%%%%%%%%%%%%%%%%%%%%%%%%%%%%%%%%%%%%%%%%%%%%%%%%%%%%%

\section{Appendix}
\label{sec:appendix}

\subsection{Large Deviation Principle and Contraction Principle}

We start by giving a formal definition of the large deviation principle. 
We refer to Dembo and Zeitouni \cite{Dembo} or Varadhan \cite{VaradhanII} 
for general background of large deviations and the applications. 

\begin{definition}[Large Deviation Principle]
A sequence $(P_{\epsilon})_{\epsilon\in\mathbb{R}^{+}}$ of probability measures on a topological space $X$ 
satisfies the large deviation principle with rate function $I:X\rightarrow\mathbb{R}$ if $I$ is non-negative, 
lower semicontinuous and for any measurable set $A$, we have
\begin{equation}
-\inf_{x\in A^{o}}I(x)\leq\liminf_{\epsilon\rightarrow 0}\epsilon\log P_{\epsilon}(A)
\leq\limsup_{\epsilon\rightarrow 0}\epsilon\log P_{\epsilon}(A)\leq-\inf_{x\in\overline{A}}I(x).
\end{equation}
Here, $A^{o}$ is the interior of $A$ and $\overline{A}$ is its closure. 
\end{definition}

The contraction principle plays a key role in our proofs. For the convenience
of the readers, we state the result as follows:

\begin{theorem}[Contraction Principle, e.g. Theorem 4.2.1. \cite{Dembo}]\label{Contraction}
If $P_{\epsilon}$ satisfies a large deviation principle on $X$ with rate 
function $I(x)$ and $F:X\rightarrow Y$ is a continuous map,
then the probability measures $Q_{\epsilon}:=P_{\epsilon}F^{-1}$ satisfies
a large deviation principle on $Y$ with rate function
\begin{equation}
J(y)=\inf_{x: F(x)=y}I(x).
\end{equation}
\end{theorem}

\subsection{Proofs of the Results in Section \ref{MainSection}}

\begin{proof}[Proof of Theorem \ref{MainThm}]
(i) We start by first proving the relation
\begin{align}\label{Th1eq1}
\lim_{T\to 0} T \log C(T) = \lim_{T\to 0} T \log \mathbb{P}
\left( \frac{1}{T} \int_0^T S_t dt \geq K \right)\,.
\end{align}
By H\"{o}lder's inequality, for any $\frac{1}{p}+\frac{1}{p'}=1$, $p,p'>1$,
\begin{align}
C(T)&\leq e^{-rT}\mathbb{E}\left[\left|\frac{1}{T}\int_{0}^{T}S_{t}dt-K\right|1_{\frac{1}{T}\int_{0}^{T}S_{t}dt\geq K}\right]
\\
&\leq e^{-rT}\left(\mathbb{E}\left[\left|\frac{1}{T}\int_{0}^{T}S_{t}dt-K\right|^{p}\right]\right)^{\frac{1}{p}}
\mathbb{P}\left(\frac{1}{T}\int_{0}^{T}S_{t}dt\geq K\right)^{\frac{1}{p'}}.
\end{align}
Let us assume that $p\geq 2$.
Note that for $p\geq 2$, $x\mapsto x^{p}$ is a convex function for $x\geq 0$
and by Jensen's inequality, $(\frac{x+y}{2})^{p}\leq\frac{x^{p}+y^{p}}{2}$ for any $x,y\geq 0$.
Therefore 
\begin{align}\label{UpI}
\mathbb{E}\left[\left|\frac{1}{T}\int_{0}^{T}S_{t}dt-K\right|^{p}\right]
&\leq\mathbb{E}\left[\left(\frac{1}{T}\int_{0}^{T}S_{t}dt+K\right)^{p}\right]
\\
&\leq 2^{p-1}\left[\mathbb{E}\left[\left(\frac{1}{T}\int_{0}^{T}S_{t}dt\right)^{p}\right]+K^{P}\right].
\nonumber
\end{align}
By Jensen's inequality again,
\begin{equation}\label{UpII}
\mathbb{E}\left[\left(\frac{1}{T}\int_{0}^{T}S_{t}dt\right)^{p}\right]
\leq\mathbb{E}\left[\frac{1}{T}\int_{0}^{T}S_{t}^{p}dt\right]
=\frac{1}{T}\int_{0}^{T}\mathbb{E}[S_{t}^{p}]dt.
\end{equation}
By It\^{o}'s formula,
\begin{equation}\label{Sp}
d(S_{t}^{p})=pS_{t}^{p-1}[(r-q)S_{t}dt+\sigma(S_{t})dW_{t}]
+\frac{1}{2}p(p-1)S_{t}^{p-2}\sigma(S_{t})^{2}S_{t}^{2}dt.
\end{equation}
Taking expectations of the both sides of the equation \eqref{Sp},
\begin{equation}
d\mathbb{E}[S_{t}^{p}]=p(r-q)\mathbb{E}[S_{t}^{p}]dt
+\frac{1}{2}p(p-1)\mathbb{E}[S_{t}^{p}\sigma(S_{t})^{2}]dt.
\end{equation}
Since $\underline{\sigma}\leq\sigma(\cdot)\leq\overline{\sigma}$, we conclude
that $\mathbb{E}[S_{t}^{p}]\leq m(t)$ where $m(t)$ is the solution to the ODE:
\begin{equation}
dm(t)=p(r-q)m(t)dt+\frac{1}{2}p(p-1)\overline{\sigma}^{2}m(t)dt,
\qquad
m(0)=S_{0}^{p},
\end{equation}
which has the solution $m(t)=S_{0}^{p}e^{(p(r-q)+\frac{1}{2}p(p-1)\overline{\sigma}^{2})t}$.
Hence,
\begin{equation}\label{UpIII}
\frac{1}{T}\int_{0}^{T}\mathbb{E}[S_{t}^{p}]dt
\leq\max_{0\leq t\leq T}m(t)
\leq S_{0}^{p}e^{|p(r-q)+\frac{1}{2}p(p-1)\overline{\sigma}^{2}|T}.
\end{equation}

Therefore, by equations \eqref{UpI}, \eqref{UpII}, \eqref{UpIII}, we have
\begin{equation}
\limsup_{T\rightarrow 0}T\log C(T)
\leq\limsup_{T\rightarrow 0}\frac{1}{p'}T\log\mathbb{P}\left(\frac{1}{T}\int_{0}^{T}S_{t}dt\geq K\right).
\end{equation}
Since it holds for any $2>p'>1$, we have the upper bound.

For any $\epsilon>0$,
\begin{align}
C(T)&\geq e^{-rT}\mathbb{E}\left[\left(\frac{1}{T}\int_{0}^{T}S_{t}dt-K\right)1_{\frac{1}{T}\int_{0}^{T}S_{t}dt\geq K+\epsilon}\right]
\\
&\geq e^{-rT}\epsilon\mathbb{P}\left(\frac{1}{T}\int_{0}^{T}S_{t}dt\geq K+\epsilon\right),
\nonumber
\end{align}
which implies that
\begin{equation}
\liminf_{T\rightarrow 0}T\log C(T)
\geq\liminf_{T\rightarrow 0}T\log\mathbb{P}\left(\frac{1}{T}\int_{0}^{T}S_{t}dt\geq K+\epsilon\right).
\end{equation}
Since it holds for any $\epsilon>0$, we get the lower bound.

So the problem boils down to compute the limit
\begin{equation}
\lim_{T\rightarrow 0}T\log\mathbb{P}\left(\frac{1}{T}\int_{0}^{T}S_{t}dt\geq K\right)
=\lim_{T\rightarrow 0}T\log\mathbb{P}\left(\int_{0}^{1}S_{tT}dt\geq K\right).
\end{equation}
Let $X_{t}:=\log S_{t}$. This is equivalent to computing the limit
\begin{equation}
\lim_{T\rightarrow 0}T\log\mathbb{P}\left(\int_{0}^{1}e^{X_{tT}}dt\geq K\right),
\end{equation}
where by It\^{o}'s lemma,
\begin{equation}
dX_{t}=\left(r-q-\frac{1}{2}\sigma^{2}(e^{X_{t}})\right)dt+\sigma(e^{X_{t}})dW_{t},\qquad X_{0}=\log S_{0}.
\end{equation}
>From the large deviations theory for small time diffusions, it was first proved
in Varadhan \cite{Varadhan67} that under the assumptions \eqref{assumpI}, \eqref{assumpII}, 
$\mathbb{P}(X_{\cdot T}\in\cdot)$ satisfies
a sample path large deviation principle on $L_{\infty}[0,1]$ with the rate function
\begin{equation}
I(g)=\frac{1}{2}\int_{0}^{1}\left(\frac{g'(t)}{\sigma(e^{g(t)})}\right)^{2}dt.
\end{equation}
with $g(0)=\log S_{0}$ and $g\in\mathcal{AC}[0,1]$, the space of absolutely continuous functions and $I(g)=+\infty$ otherwise.

Note that the map $g\mapsto\int_{0}^{1}e^{g(x)}dx$ from $L_{\infty}[0,1]$ to $\mathbb{R}^{+}$ is a continuous map.
Therefore, by contraction principle (Theorem \ref{Contraction}),
$\mathbb{P}(\int_{0}^{1}e^{X_{tT}}dt\in\cdot)$ satisfies a large deviation principle
with the rate function
\begin{equation}
\mathcal{I}(x,S_{0}):=\inf_{
\substack{\int_{0}^{1}e^{g(t)}dt=x
\\
g(0)=\log S_{0}, g\in\mathcal{AC}[0,1]}}\frac{1}{2}\int_{0}^{1}\left(\frac{g'(t)}{\sigma(e^{g(t)})}\right)^{2}dt.
\end{equation}
Hence, for out-of-the money call options, i.e. $S_{0}<K$,
\begin{equation}
\lim_{T\rightarrow 0}T\log\mathbb{P}\left(\frac{1}{T}\int_{0}^{T}S_{t}dt\geq K\right)
=-\inf_{x\geq K}\mathcal{I}(x,S_{0})=-\mathcal{I}(K,S_{0}),
\end{equation}
where the last step is due to the fact that $\mathcal{I}(K,S_{0})$ is
increasing in $K$ for $K>S_{0}$, see Proposition \ref{prop:2}.

(ii) We conclude by proving the analogous relation to (\ref{Th1eq1}) for Asian put options, 
that is, for out-of-the-money put options, $S_{0}>K$, we will show that
\begin{equation}
\lim_{T\rightarrow 0}T\log\mathbb{P}\left(K\geq\frac{1}{T}\int_{0}^{T}S_{t}dt\right)
=-\mathcal{I}(K,S_{0}).
\end{equation}

By H\"{o}lder's inequality, for any $\frac{1}{p}+\frac{1}{p'}=1$, $p,p'>1$,
\begin{align}
P(T)&=e^{-rT}\mathbb{E}\left[\left(K-\frac{1}{T}\int_{0}^{T}S_{t}dt\right)^{+}1_{K\geq\frac{1}{T}\int_{0}^{T}S_{t}dt}\right]
\\
&\leq e^{-rT}\left(\mathbb{E}\left[\left(\left(K-\frac{1}{T}\int_{0}^{T}S_{t}dt\right)^{+}\right)^{p}\right]\right)^{\frac{1}{p}}
\mathbb{P}\left(K\geq\frac{1}{T}\int_{0}^{T}S_{t}dt\right)^{\frac{1}{p}}
\nonumber
\\
&\leq e^{-rT}K\mathbb{P}\left(K\geq\frac{1}{T}\int_{0}^{T}S_{t}dt\right)^{\frac{1}{p'}}.
\nonumber
\end{align}
Thus, $\limsup_{T\rightarrow 0}T\log P(T)\leq-\frac{1}{p'}\mathcal{I}(K,S_{0})$. Since it holds for any $p'>1$, we proved the upper bound.

For the lower bound, for any sufficiently small $\epsilon>0$,
\begin{align}
P(T)&\geq e^{-rT}\mathbb{E}\left[\left(K-\frac{1}{T}\int_{0}^{T}S_{t}dt\right)1_{K\geq\frac{1}{T}\int_{0}^{T}S_{t}dt+\epsilon}\right]
\\
&\geq e^{-rT}\epsilon\mathbb{P}\left(K\geq\frac{1}{T}\int_{0}^{T}S_{t}dt+\epsilon\right),
\nonumber
\end{align}
which implies that $\liminf_{T\rightarrow 0}T\log P(T)\geq-\mathcal{I}(K-\epsilon,S_{0})$. By letting $\epsilon\downarrow 0$,
and the fact that the rate function $\mathcal{I}(K,S_{0})$ is continuous in $K$, see Proposition \ref{prop:1}
and Remark \ref{CtsRemark}, we proved the lower bound.
\end{proof}

\begin{proof}[Proof of Corollary \ref{InMoneyCor}]
(i) From the put-call parity,
\begin{align}
C(T)-P(T)&=e^{-rT}\mathbb{E}\left[\frac{1}{T}\int_{0}^{T}S_{t}dt-K\right]
\\
&=e^{-rT}\left[\frac{1}{T}\int_{0}^{T}e^{(r-q)t}S_{0}dt-K\right]
\nonumber
\\
&=
\begin{cases}
e^{-rT}[S_{0}-K] &\text{if $r=q$}
\\
e^{-rT}\left[\frac{1}{T(r-q)}[e^{(r-q)T}-1]S_{0}-K\right] &\text{if $r\neq q$}
\end{cases}
\nonumber
\\
&=
\begin{cases}
[S_{0}-K](1-rT)+O(T^{2}) &\text{if $r=q$}
\\
S_{0}-K-\frac{1}{2}(r+q)S_{0}T+KrT+O(T^{2}) &\text{if $r\neq q$}
\end{cases}
\nonumber,
\end{align}
as $T\rightarrow 0$. Therefore, for in-the-money call option, i.e. $S_{0}>K$,
from Theorem \ref{MainThm}, we get $C(T)=S_{0}-K-\frac{1}{2}(r+q)S_{0}T+KrT+O(T^{2})$ as $T\rightarrow 0$.

(ii) For in-the-money put option, i.e. $S_{0}<K$, from (i) and Theorem \ref{MainThm}, we get
$P(T)=K-S_{0}+\frac{1}{2}(r+q)S_{0}T-KrT+O(T^{2})$ as $T\rightarrow 0$.
\end{proof}

\begin{proof}[Proof of Theorem \ref{ATMThm}]
(i) For at-the-money call option,
\begin{equation}
C(T)=e^{-rT}\mathbb{E}\left[\left(\frac{1}{T}\int_{0}^{T}e^{(r-q)t}X_{t}dt-S_{0}\right)^{+}\right],
\end{equation}
where $X_{t}:=\frac{S_{t}}{e^{(r-q)t}}$ is a martingale and satisfies the SDE:
\begin{equation}
dX_{t}=\sigma(X_{t}e^{(r-q)t})X_{t}dW_{t},
\end{equation}
with $X_{0}=S_{0}$.

\textbf{Claim 1}. As $T\rightarrow 0$,
\begin{equation}\label{ClaimI}
\left|\mathbb{E}\left[\left(\frac{1}{T}\int_{0}^{T}e^{(r-q)t}X_{t}dt-S_{0}\right)^{+}\right]
-\mathbb{E}\left[\left(\frac{1}{T}\int_{0}^{T}X_{t}dt-S_{0}\right)^{+}\right]\right|=O(T).
\end{equation}
Let us prove \eqref{ClaimI}.
\begin{align}
&\left|\mathbb{E}\left[\left(\frac{1}{T}\int_{0}^{T}e^{(r-q)t}X_{t}dt-S_{0}\right)^{+}\right]
-\mathbb{E}\left[\left(\frac{1}{T}\int_{0}^{T}X_{t}dt-S_{0}\right)^{+}\right]\right|
\\
&\leq
\mathbb{E}\left[\left|\left(\frac{1}{T}\int_{0}^{T}e^{(r-q)t}X_{t}dt-S_{0}\right)^{+}
-\left(\frac{1}{T}\int_{0}^{T}X_{t}dt-S_{0}\right)^{+}\right|\right]
\nonumber
\\
&\leq
\mathbb{E}\left[\frac{1}{T}\int_{0}^{T}\left|e^{(r-q)t}-1\right|X_{t}dt\right]
\nonumber
\\
&=S_{0}\frac{1}{T}\int_{0}^{T}|e^{(r-q)t}-1|dt
\nonumber
\\
&=S_{0}\left|\frac{1}{T}\int_{0}^{T}(e^{(r-q)t}-1)dt\right|
\nonumber
\\
&=S_{0}\left|\frac{e^{(r-q)T}-1}{(r-q)T}-1\right|.
\nonumber
\end{align}
Hence, we proved \eqref{ClaimI}.

Next, let us define $\hat{X}_{t}$, which satisfies the SDE:
\begin{equation}
d\hat{X}_{t}=\sigma(S_{0})S_{0}dW_{t},
\qquad
\hat{X}_{0}=S_{0}.
\end{equation}

\textbf{Claim 2}. 
\begin{equation}\label{ClaimII}
\mathbb{E}\left[\max_{0\leq t\leq T}|X_{t}-\hat{X}_{t}|\right]=O(T),\qquad\text{as $T\rightarrow 0$}.
\end{equation}
Let us prove \eqref{ClaimII}. Note that
\begin{equation}
X_{t}-\hat{X}_{t}
=\int_{0}^{t}\left[\sigma(X_{s}e^{(r-q)s})X_{s}-\sigma(S_{0})S_{0}\right]dW_{s}.
\end{equation}
By It\^{o}'s isometry and the uniform Lipschitz assumption,
\begin{align}\label{est}
&\mathbb{E}[(X_{t}-\hat{X}_{t})^{2}]
\\
&=\int_{0}^{t}\mathbb{E}\left[\left(\sigma(X_{s}e^{(r-q)s})X_{s}-\sigma(S_{0})S_{0}\right)^{2}\right]ds
\nonumber
\\
&\leq 
2\int_{0}^{t}\mathbb{E}\left[\left(\sigma(X_{s})X_{s}-\sigma(S_{0})S_{0}\right)^{2}\right]ds
+2\int_{0}^{t}\mathbb{E}\left[\left(\sigma(X_{s}e^{(r-q)s})X_{s}-\sigma(X_{s})X_{s}\right)^{2}\right]ds
\nonumber
\\
&\leq 2\alpha^{2}\int_{0}^{t}\mathbb{E}[(X_{s}-S_{0})^{2}]ds
+2\beta^{2}\int_{0}^{t}(e^{(r-q)s}-1)^{2}\mathbb{E}[X_{s}^{4}]ds
\nonumber
\\
&\leq 4\alpha^{2}\int_{0}^{t}\mathbb{E}[(X_{s}-\hat{X}_{s})^{2}]ds
+4\alpha^{2}\int_{0}^{t}\mathbb{E}[(\hat{X}_{s}-S_{0})^{2}]ds
+2\beta^{2}\int_{0}^{t}(e^{(r-q)s}-1)^{2}\mathbb{E}[X_{s}^{4}]ds.
\nonumber
\end{align}
Note that $\hat{X}_{s}=S_{0}+\sigma(S_{0})S_{0}W_{s}$ and hence
$\mathbb{E}[(\hat{X}_{s}-S_{0})^{2}]=\sigma^2(S_{0}) S_{0}^{2}s$.
Also
\begin{equation}
|\sigma(x)x|=|\sigma(x)x-0\sigma(0)|\leq\alpha|x|
\end{equation}
so that $\sigma(x)\leq\alpha$.

Since $X_{t}-X_{0}$ is a martingale starting at $0$, by the
Burkholder-Davis-Gundy inequality, we get
\begin{equation*}
\mathbb{E}[(X_{t}-X_{0})^{4}]
\leq C\mathbb{E}[(\langle X\rangle_{t})^{2}], 
\end{equation*}
for some constant $C>0$, where $\langle X\rangle_{t}=\int_{0}^{t}\sigma^{2}(X_{s}e^{(r-q)s})X_{s}^{2}ds$
is the quadratic variation of $X_{t}$. Using the Cauchy-Schwarz inequality, we get
\begin{equation*}
\mathbb{E}[(X_{t}-X_{0})^{4}]
\leq C\mathbb{E}\left[\left(\int_{0}^{t}\sigma^{2}(X_{s}e^{(r-q)s})X_{s}^{2}ds\right)^{2}\right]
\leq C\mathbb{E}\left[\int_{0}^{t}\sigma^{4}(X_{s}e^{(r-q)s})ds\int_{0}^{t}X_{s}^{4}ds\right].
\end{equation*}
Since $\sigma(\cdot)\leq\alpha$ and $\mathbb{E}[X_{t}^{4}]\leq 8X_{0}^{4}+8\mathbb{E}[(X_{t}-X_{0})^{4}]$
(since $(\frac{x+y}{2})^{4}\leq\frac{x^{4}+y^{4}}{2}$ for any $x,y\geq 0$)
and $X_{0}=S_{0}$, we get, for any sufficiently small $t$, say $t\leq 1$,
\begin{equation*}
\mathbb{E}[X_{t}^{4}]
\leq 8S_{0}^{4}+8C\alpha^{4}\int_{0}^{t}\mathbb{E}[X_{s}^{4}]ds,
\end{equation*}
Gronwall's inequality states that if $\beta(\cdot)$ is non-negative and for any $t\geq 0$,
$u(t)\leq\alpha(t)+\int_{0}^{t}\beta(s)u(s)ds$, then
\begin{equation}
u(t)\leq\alpha(t)+\int_{0}^{t}\alpha(s)\beta(s)e^{\int_{s}^{t}\beta(r)dr}ds.
\end{equation}
Therefore, by Gronwall's inequality we have for any sufficiently small $t$, 
\begin{equation*}
\mathbb{E}[X_{t}^{4}]\leq 8S_{0}^{4}e^{8C\alpha^{4}t}.
\end{equation*}
Hence, there exists a constant $\gamma>0$ so that
\begin{equation}
2\beta^{2}\int_{0}^{t}(e^{(r-q)s}-1)^{2}\mathbb{E}[X_{s}^{4}]ds
\leq 16\beta^{2}S_{0}^4 e^{8C\alpha^{4}t}\int_{0}^{t}(e^{(r-q)s}-1)^{2}ds
\leq\gamma t^{2},
\end{equation}
for any sufficiently small $t>0$.
Plugging into \eqref{est}, we get
\begin{equation}
\mathbb{E}[(X_{t}-\hat{X}_{t})^{2}]
\leq 4\alpha^{2}\int_{0}^{t}\mathbb{E}[(X_{s}-\hat{X}_{s})^{2}]ds
+2\alpha^{2}\sigma^2(S_{0}) S_{0}^{2}t^{2}
+\gamma t^{2}.
\end{equation}
By Gronwall's inequality, we have
\begin{align}
\mathbb{E}[(X_{t}-\hat{X}_{t})^{2}]
&\leq[2\alpha^{2}\sigma^2(S_{0}) S_{0}^{2}+\gamma]t^{2}
+4\alpha^{2}\int_{0}^{t}[2\alpha^{2}\sigma^2 (S_{0}) S_{0}^{2}+\gamma]s^{2}e^{4\alpha^{2}(t-s)}ds
\\
&\leq[2\alpha^{2}\sigma^2(S_{0}) S_{0}^{2}+\gamma]\left[t^{2}+\frac{4}{3}\alpha^{2}t^{3}e^{4\alpha^{2}t}\right].
\nonumber
\end{align}
Hence, we conclude that there exists some universal constant $M>0$, so that for any sufficiently small $T>0$,
\begin{equation}
\mathbb{E}[(X_{T}-\hat{X}_{T})^{2}]\leq MT^{2}.
\end{equation}
Finally notice that $X_{t}-\hat{X}_{t}$ is a martingale since both $X_{t}$ and $\hat{X}_{t}$ are.
By Doob's martingale inequality, for sufficiently small $T>0$,
\begin{equation}
\mathbb{E}\left[\max_{0\leq t\leq T}|X_{t}-\hat{X}_{t}|\right]
\leq 2\left(\mathbb{E}[(X_{T}-\hat{X}_{T})^{2}]\right)^{\frac{1}{2}}
\leq 2\sqrt{M}T.
\end{equation}
Hence, we proved \eqref{ClaimII}.

\textbf{Claim 3}. For $T\rightarrow 0$,
\begin{equation}\label{ClaimIII}
\left|\mathbb{E}\left[\left(\frac{1}{T}\int_{0}^{T}X_{t}dt-S_{0}\right)^{+}\right]
-\mathbb{E}\left[\left(\frac{1}{T}\int_{0}^{T}\hat{X}_{t}dt-S_{0}\right)^{+}\right]\right|
=O(T).
\end{equation}
Let us prove \eqref{ClaimIII}. 
\begin{align}
&\left|\mathbb{E}\left[\left(\frac{1}{T}\int_{0}^{T}X_{t}dt-S_{0}\right)^{+}\right]
-\mathbb{E}\left[\left(\frac{1}{T}\int_{0}^{T}\hat{X}_{t}dt-S_{0}\right)^{+}\right]\right|
\\
&\leq\mathbb{E}\left|\frac{1}{T}\int_{0}^{T}X_{t}dt-\frac{1}{T}\int_{0}^{T}\hat{X}_{t}dt\right|
\nonumber
\\
&\leq\mathbb{E}\left[\max_{0\leq t\leq T}|X_{t}-\hat{X}_{t}|\right].
\nonumber
\end{align}
Therefore, \eqref{ClaimIII} follows from \eqref{ClaimII}.

\textbf{Claim 4}. 
\begin{equation}\label{ClaimIV}
\mathbb{E}\left[\left(\frac{1}{T}\int_{0}^{T}\hat{X}_{t}dt-S_{0}\right)^{+}\right]
=\sigma(S_{0})S_{0}\frac{\sqrt{T}}{\sqrt{3}}\mathbb{E}[Z1_{Z>0}],
\end{equation}
where $Z\sim N(0,1)$.
Let us prove \eqref{ClaimIV}. Note that $\hat{X}_{t}=S_{0}+\sigma(S_{0})S_{0}W_{t}$.
Hence, 
\begin{equation}
\frac{1}{T}\int_{0}^{T}\hat{X}_{t}dt-S_{0}
=\sigma(S_{0})S_{0}\frac{1}{T}\int_{0}^{T}W_{t}dt
\sim N\left(0,\sigma^2(S_{0}) S_{0}^{2}\frac{T}{3}\right).
\end{equation}
Hence, we proved \eqref{ClaimIV}. Indeed, we can compute that
\begin{equation}\label{ClaimV}
\mathbb{E}\left[Z1_{Z>0}\right]
=\frac{1}{\sqrt{2\pi}}\int_{0}^{\infty}xe^{-\frac{x^{2}}{2}}dx
=\frac{1}{\sqrt{2\pi}}.
\end{equation}
Finally, putting \eqref{ClaimI}, \eqref{ClaimII}, \eqref{ClaimIII}, \eqref{ClaimIV}, and \eqref{ClaimV} together, we proved the desired result.

(ii) For at-the-money put option, the proof is similar to (i) and is omitted.
\end{proof}

\subsection{Proofs of the Results in Section \ref{Sec:3}}

\begin{proof}[Proof of Proposition \ref{VarProp}]

We present here the proof of the solution of the variational problem
in Proposition \ref{VarProp} of the variational problem (\ref{Idef}).
This variational problem can be simplified by introducing the function $f(t)$
defined as $g(t) = \log S_0 + f(t)$. Expressed in terms of this function, 
the variational problem is stated as:
\begin{equation}\label{Idef}
\mathcal{I}(K,S_0) = \inf_f \frac12 \int_0^1  \left(
\frac{f'(t)}{\sigma(S_0 e^{f(t)})}\right)^2 dt,
\end{equation}
where $f(t)\in\mathcal{AC}[0,1]$ satisfies $f(0)=0$ and 
\begin{equation}\label{const2}
\int_0^1  e^{f(t)} dt= \frac{K}{S_0} \,.
\end{equation}

The constraint (\ref{const2}) can be taken into account by introducing
a Lagrange multiplier $\lambda$ and considering the variational problem
for the functional
\begin{equation}\label{Lambdadef}
\Lambda[f] = \frac12 
\int_0^1  \left(
\frac{f'(t)}{\sigma(S_0 e^{f(t)})}\right)^2 dt+ \lambda
\left( \int_0^1  e^{f(t)} dt- \frac{K}{S_0} \right)\,.
\end{equation}
with boundary condition $f(0)=0$.

First we recall a well-known result in the calculus of variations, see 
Sec.~IV.5 in \cite{CR1}, stated in a form which is representative
for the variational problems encountered in this paper.
\begin{lemma}\label{Lemma1}
Consider the variational problem of finding the extremum of the functional
\begin{equation}
\Lambda[x] = 
\frac12 \int_0^T \left(\frac{x'(t)}{\Sigma(x(t))}\right)^2 dt
 - \int_0^T V(x(t)) dt - f(x(T)),
\end{equation}
where $\Sigma(x), V(x)$ are $C^1$ functions, over the set of functions $x(t)$ 
satisfying the constraint
\begin{equation}\label{x0}
x(0) = x_0 \, .
\end{equation}
The optimal function $x(t)$ satisfies the Euler-Lagrange equation 
\begin{equation}\label{ELeq0}
\frac{d}{dt}\left( \frac{x'(t)}{\Sigma(x(t))}\right) = - V'(x(t)) \Sigma(x(t))
\end{equation}
with boundary condition (\ref{x0}) at $t=0$ and the transversality condition 
\begin{equation}\label{transv}
x'(T) = f(x(T)) \Sigma^2(x(T))
\end{equation}
at $t=T$.
\end{lemma}

\begin{proof}
Define $x_\epsilon(t)$ as a perturbation around the optimal function $x(t)$ 
\begin{equation}
x_\epsilon(t) = x(t) + \epsilon \eta(t)\, ,
\end{equation}
where $\eta(t)\in C^1$ is a function which vanishes at $t=0$ 
but is unconstrained at $t=T$
\begin{equation}\label{eta0}
\eta(0) = 0\,.
\end{equation}
A necessary condition that $\Lambda[x]$ has an extremum on $x(t)$ is that
we have
\begin{eqnarray}
\frac{d}{d\epsilon} \Lambda[x_\epsilon] |_{\epsilon=0} &=& 
\int_0^T \frac{x'(t) }{\Sigma^2(x(t))} \eta'(t) dt \\
& & - 
\int_0^T \left( \frac{[x'(t)]^2}{\Sigma^3(x(t))} \Sigma'(x(t)) 
+ V'(x(t)) \right) \eta(t) dt - f'(x(T)) \eta(T)  \nonumber \\
 &=& \int_0^T \eta(t) \left\{
- \frac{d}{dt} \left( \frac{x'(t)}{\Sigma^2(x(t))} \right) - 
\frac{[x'(t)]^2}{\Sigma^2(x(t))} \frac{\Sigma'(x(t))}{\Sigma(x(t))} - V'(x(t)) \right\} dt\nonumber\\
&+& \eta(T) \left( \frac{x'(T)}{\Sigma^2(x(T))} - f'(x(T)) \right) = 0
\nonumber
\end{eqnarray}
for any $\eta(t)$ satisfying the constraint (\ref{eta0}). 
In the second equality we integrated by parts in the first term.
At optimality the functional derivative $\frac{d}{d\epsilon}\Lambda[x_\epsilon]$ 
must vanish for any $\eta(t)$ satisfying the constraint $\eta(0)=0$.
This requires that the two terms vanish separately for any $\eta(t)$.
The vanishing of the first term gives the Euler-Lagrange 
equation (\ref{ELeq0}), and the vanishing of the second term gives the 
transversality condition (\ref{transv}).
\end{proof}

Application of this result to the variational problem (\ref{Lambdadef})
gives that the optimal function $f$ satisfies the Euler-Lagrange 
equation 
\begin{equation}\label{ELsimple}
\frac{d}{dt} \left( \frac{f'(t)}{\sigma(S_0 e^{f(t)})} \right) =
\lambda e^{f(t)} \sigma(S_0 e^{f(t)})\,,
\end{equation}
and the transversality condition
\begin{equation}
f'(1)=0.
\end{equation}

It is known \cite{CR1} that one can relax the conditions $\Sigma(x),
V(x) \in C^1$ to allow functions which are piecewise continuous. This is made 
explicit by writing the Euler-Lagrange equation (\ref{ELsimple}) in an 
alternative
form as given by Lemma~\ref{Lemma2}. Note that the derivative $\Sigma'(x)$
does not appear in this result anymore. This implies that the conditions 
(\ref{assumpI}),
(\ref{assumpII}) on $\sigma(\cdot)$ are sufficient for our purposes.

\begin{lemma}\label{Lemma2}
The Euler-Lagrange equation (\ref{ELsimple}) can be written alternatively as
\begin{equation}\label{Edef}
\frac12 \left( \frac{f'(t)}{\sigma(S_0 e^{f(t)})} \right)^2 
- \lambda e^{f(t)} = c
\end{equation}
with $c$ a constant.
\end{lemma}

\begin{proof} 
Follows by multiplying both sides of (\ref{ELsimple}) with 
$f'(t)/\sigma(S_0 e^{f(t)})$
\begin{equation}
\frac12 \frac{d}{dt} \left( \frac{f'(t)}{\sigma(S_0 e^{f(t)})} \right)^2 =
 \lambda e^{f(t)} f'(t) =  \lambda \frac{d}{dt} e^{f(t)} \,,
\end{equation}
This reproduces (\ref{Edef}).
\end{proof}

>From (\ref{Edef}) and the transversality condition $f'(1)=0$ we have
\begin{equation}\label{Edef2}
\frac12 \left( \frac{f'(t)}{\sigma(S_0 e^{f(t)})} \right)^2 
- \lambda e^{f(t)} = -\lambda e^{f(1)}.
\end{equation}
Integrating over $t:(0,1)$ this gives
\begin{equation}\label{13}
\mathcal{I}(K,S_0) = \frac12 \int_0^1 \left( \frac{f'(t)}{\sigma(S_0 e^{f(t)})} \right)^2 dt
= - \lambda e^{f(1)} + \lambda \int_0^1 e^{f(t)} dt= 
 \lambda( K/S_0 - e^{f(1)}) \,.
\end{equation}
This relation expresses the rate function $\mathcal{I}(K,S_0)$ in terms of the 
terminal value of the optimal function $f(1)$ and the Lagrange multiplier $\lambda$. 

The optimal function $f(t)$, the solution of the variational problem (\ref{ELsimple}),
has the following qualitative behavior:

1. $f(t)$ is increasing $f'(t)>0$ for $\lambda<0$. This case corresponds to
$K>S_0$. For this case $f(1)>0$.

2. $f(t)$ is decreasing $f'(t)<0$ for $\lambda>0$. This case corresponds to 
$K < S_0$. For this case $f(1)<0$.

These properties follow from the Euler-Lagrange equation (\ref{ELsimple}).
Integrating this equation over $t:(t,1)$ and using the transversality condition
$f'(1)=0$ gives
\begin{equation}
f'(t) = - \lambda \Sigma(f(t)) \int_t^1 e^{f(s)} \Sigma(f(s)) dt.
\end{equation}
The factors multiplying $\lambda$ are negative, such that the sign of this
expression is opposite to that of $\lambda$. 
The integrand in the constraint (\ref{const2}) satisfies the inequality
$e^{f(t)} > e^{f(0)}= 1$ in the case 1.
and $e^{f(t)} < e^{f(0)} =1$ in case 2. These two cases correspond to $K>S_0$ and 
$K<S_0$, respectively.

We can use (\ref{Edef}) to eliminate $f'(t)$ in terms of $f(t)$ as
\begin{equation}\label{31}
f'(t) = 
\begin{cases}
\sqrt{-2\lambda} \Sigma(f(t)) \sqrt{e^{f(1)} - e^{f(t)}} \,, & K>S_0\,, \lambda<0  
\\
-\sqrt{2\lambda} \Sigma(f(t)) \sqrt{e^{f(t)} - e^{f(1)}} \,, & K<S_0\,, \lambda>0
\end{cases}.
\end{equation}

We will treat the two cases separately. 

{\bf Case 1.} $K > S_0$. We will show that the rate 
function is
\begin{equation}\label{Icase1}
\mathcal{I}(K,S_0) = 
%\lambda\left( \frac{K}{S_0} - e^{f(1)}\right) = 
\frac12 F^{(-)}(f_1) G^{(-)}(f_1),
\end{equation}
where $f_1 > 0$ is the solution of the equation
\begin{equation}
e^{f_1} - K/S_0 = \frac{G^{(-)}(f_1)}{F^{(-)}(f_1)},
\end{equation}
with
\begin{align}
G^{(-)}(f_1) &= \int_{0}^{f_1} \frac{1}{\Sigma(y)} \sqrt{e^{f_1} - e^y}dy, \\
F^{(-)}(f_1) &= \int_{0}^{f_1} \frac{1}{\Sigma(y)} 
\frac{1}{\sqrt{e^{f_1} - e^y}}dy \,.
\end{align}

{\bf Proof.} First we note that the rate function can be written in two
equivalent ways as
\begin{equation}
\mathcal{I}(K,S_0) = \lambda( K/S_0 - e^{f_1}) = 
\sqrt{-\frac{\lambda}{2}} \int_{0}^{f_1}
\frac{1}{\Sigma(y)} \sqrt{e^{f_1} - e^y}dy.
\end{equation}
The first equality is just equation (\ref{13}) with $f_1=f(1)$. 
The second equality follows by changing the integration variable in the 
definition of the rate function from $t$ to $f(t)$
\begin{align}
\mathcal{I}(K,S_0) &=
\frac12 \int_0^1 \left( \frac{f'(t)}{\sigma(S_0 e^{f(t)})} \right)^2 dt
\\
&=\frac12 \int_{0}^{f_1} \frac{1}{f'(t)} 
\left( \frac{f'(t)}{\sigma(S_0 e^{f(t)})} \right)^2 df
\nonumber
\\
&=\sqrt{-\frac{\lambda}{2}}
\int_{0}^{f_1} \frac{1}{\Sigma(y)} \sqrt{e^{f_1} - e^y} dy
\equiv \sqrt{-\frac{\lambda}{2}} G^{(-)}(f_1) \, ,\nonumber
\end{align}
where we used in the second line the relation (\ref{31}) to eliminate $f'(t)$
in the numerator and denominator.

We need a second equation for $(\lambda, f_1)$ in order to be able to eliminate
$\lambda$. This is obtained by writing
\begin{align}\label{21}
1 &= \int_0^1 dt = \int_{0}^{f_1} \frac{1}{f'(t)} df\\
&=
\frac{1}{\sqrt{-2\lambda}} \int_{0}^{f_1} \frac{1}{\Sigma(y)} 
\frac{1}{\sqrt{e^{f_1} - e^y}} dy\equiv \frac{1}{\sqrt{-2\lambda}}  
F^{(-)}(f_1)\,. \nonumber
\end{align}

Eliminating $\lambda$ between these equations we get the result (\ref{Icase1}).

{\bf Case 2.} $K < S_0$. We will show that the rate function is
\begin{equation}\label{Icase2}
\mathcal{I}(K,S_0) = 
\frac12 F^{(+)}(h_1) G^{(+)}(h_1),
\end{equation}
where $h_1 = -f(1) >0$ is given by the solution of the equation
\begin{equation}\label{138}
\frac{K}{S_0} - e^{-h_1}  = \frac{G^{(+)}(h_1)}{F^{(+)}(h_1)},
\end{equation}
with
\begin{align}
G^{(+)}(h_1) &= \int_{0}^{h_1} \frac{1}{\Sigma(-y)} 
\sqrt{e^{-y} - e^{-h_1}}dy,\\
F^{(+)}(h_1) &= \int_{0}^{h_1} \frac{1}{\Sigma(-y)} 
\frac{1}{\sqrt{e^{-y} - e^{-h_1}}}dy \,.
\end{align}

{\bf Proof.} 
The proof follows the same approach as in the previous case,
using the relation (\ref{31}) to eliminate $f'(t)$. 
The rate function is
\begin{align}
\mathcal{I}(K,S_0)&=
\frac12 \int_0^1 \left( \frac{f'(t)}{\sigma(S_0 e^{f(t)})} \right)^2 dt
\\
&=\frac12 \int_{0}^{f_1} \frac{1}{f'(t)} 
\left( \frac{f'(t)}{\sigma(S_0 e^{f(t)})} \right)^2 df
\nonumber
\\
&=-\sqrt{\frac{\lambda}{2}}
\int_{0}^{f_1} \frac{1}{\Sigma(y)} \sqrt{e^y-e^{f_1}} dy
\nonumber
\\
&=\sqrt{\frac{\lambda}{2}} 
\int_{0}^{h_1} \frac{1}{\Sigma(-y)} \sqrt{e^{-y}-e^{-h_1}} dy
\nonumber
\\
&\equiv  
\sqrt{\frac{\lambda}{2}} G^{(+)}(h_1) \,.
\nonumber 
\end{align}
A second equation is derived in a similar way to (\ref{21}) and reads
\begin{align}\label{21xx}
1 &= \int_0^1 dt = \int_{0}^{f_1} \frac{1}{f'(t)} df =
-\frac{1}{\sqrt{2\lambda}} \int_{0}^{f_1} \frac{1}{\Sigma(y)} 
\frac{1}{\sqrt{e^y-e^{f_1}}} dy\\
&=  \frac{1}{\sqrt{2\lambda}} \int_{0}^{h_1} \frac{1}{\Sigma(-y)} 
\frac{1}{\sqrt{e^{-y}-e^{-h_1}}}dy
\equiv \frac{1}{\sqrt{2\lambda}}  F^{(+)}(h_1) \,. \nonumber
\end{align}
Eliminating $\lambda$ between these equations we get the result (\ref{Icase2}).
Equation (\ref{138}) follows from (\ref{13}) by writing
\begin{align}
\mathcal{I}(K,S_0) = \lambda\left( \frac{K}{S_0} - e^{-h_1} \right) =
\frac12 F^{(+)}(h_1) G^{(+)}(h_1) \, ,
\end{align}
and using $\lambda = \frac12 (F^{(+)}(h_1))^2$.
\end{proof}

%%%%%%%%%%%%%%%%%%%%%%%%%%%%%%%%%%%%%%%%%%%%%%%%%%%%%%%%%%

\begin{proof}[Proof of Proposition \ref{prop:1}]
(i) $K > S_0$. The condition for the minimum of the function in (\ref{J1}) is
\begin{equation}\label{J1min}
\frac{d}{d\varphi} 
\frac{(\mathcal{G}^{(-)}(\varphi))^2}{\varphi - \frac{K}{S_0}} = 
2 \mathcal{G}^{(-)}(\varphi) \frac12 \mathcal{F}^{(-)}(\varphi) 
\frac{1}{\varphi - \frac{K}{S_0}} -
(\mathcal{G}^{(-)}(\varphi))^2 \frac{1}{(\varphi - \frac{K}{S_0})^2} = 0 \, ,
\end{equation}
where we introduced
\begin{equation}
\frac{d}{d\varphi} \mathcal{G}^{(-)}(\varphi) = 
\frac12 \mathcal{F}^{(-)}(\varphi) \, ,
\end{equation}
with
\begin{equation}\label{F1def}
\mathcal{F}^{(-)}(\varphi) = \int_1^{\varphi} \frac{1}{z\sigma(S_0 z)} 
\frac{1}{\sqrt{\varphi-z}}dz
\,,\quad
\varphi \geq 1\,.
\end{equation}

The equation (\ref{J1min}) gives the value $\varphi_1$ at which the function
has an extremal value
\begin{equation}\label{J1eq}
\varphi_1 - \frac{K}{S_0} = 
\frac{\mathcal{G}^{(-)}(\varphi_1)}{\mathcal{F}^{(-)}(\varphi_1)} \,.
\end{equation}
This is identical to equation \eqref{eqf1}, with the identification
$\varphi_1 = e^{f_1}$, and $\mathcal{G}^{(-)}(\varphi_1) = G^{(-)}(f_1)$,
$\mathcal{F}^{(-)}(\varphi_1) = F^{(-)}(f_1)$.

Substituting the equation (\ref{J1eq}) into (\ref{J1}) we obtain the rate 
function
\begin{equation}
\mathcal{I}(K,S_0) = 
\frac12 \mathcal{G}^{(-)}(\varphi_1) \mathcal{F}^{(-)}(\varphi_1) \,.
\end{equation}
This agrees with the result \eqref{Iresult} for the rate function for $K\geq S_0$.

(ii) $K < S_0$. The condition for the extremum of the function in (\ref{J2}) 
reads
\begin{equation}\label{J2eq}
\frac{K}{S_0} - \chi_1 = 
\frac{\mathcal{G}^{(+)}(\chi_1)}{\mathcal{F}^{(+)}(\chi_1)}\, ,
\end{equation}
where we defined
\begin{equation}
\frac{d}{d\chi} \mathcal{G}^{(+)}(\chi) = - \frac12 \mathcal{F}^{(+)}(\chi)\, ,
\end{equation}
with
\begin{equation}\label{F2def}
\mathcal{F}^{(+)}(\chi) = 
\int_{\chi}^1 \frac{1}{z\sigma(S_0 z)} \frac{1}{\sqrt{z-\chi}}dz\,,\quad
0 < \chi \leq 1\,.
\end{equation}

The equation (\ref{J2eq}) is the same as the equation \eqref{eqh1},
with the identification $\chi_1  = e^{-h_1}$, and
$\mathcal{G}^{(+)}(\chi_1) = G^{(+)}(h_1)$,
$\mathcal{F}^{(+)}(\chi_1) = F^{(+)}(h_1)$.
Substituting this equation into (\ref{J2}) we obtain the rate function as 
\begin{equation}
\mathcal{I}(K,S_0 ) = \frac12 \mathcal{G}^{(+)}(\chi_1) \mathcal{F}^{(+)}(\chi_1)\, ,
\end{equation}
which agrees with the result in \eqref{Iresult} for $K\leq S_{0}$.
\end{proof}
%%%%%%%%%%%%%%%%%%%%%%%%%%%%%%%%%%%%%%%%%%%%%%%%%%%%%%%%

\begin{proof}[Proof of Proposition \ref{prop:2}]
The proof is based on Lemma \ref{lemma:1} and Lemma \ref{lemma:2}.

i) $K>S_0$.
We will use the representation (\ref{J1}) to show that the
rate function $\mathcal{I}(K,S_0)$ is a monotonically increasing function
of $K$ for $K > S_0$. First we prove that the function 
$(\mathcal{G}^{(-)}(\varphi))^2$ 
satisfies the technical conditions required for $f(x)$ appearing in
Lemma~\ref{lemma:1}: it is positive and increasing for $\varphi>1$,
and has superlinear growth in $\varphi$ as $\varphi \to \infty$. The positivity
condition follows from the definition, and the increasing property follows
from the positivity of the integral $\mathcal{F}^{(-)}(\varphi)$ defined by 
(\ref{F1def}).

We prove next the superlinear growth condition as $\varphi \to \infty$. 
In the Black-Scholes model the function $\mathcal{G}^{(-)}(\varphi)$ can be
computed exactly and is given by
\begin{equation}
\mathcal{G}^{(-)}_{\rm BS}(\varphi) = 
\frac{1}{\sigma} \int_1^\varphi \frac{1}{z} \sqrt{\varphi-z}dz
= \frac{2}{\sigma} \left(
\sqrt{\varphi} \mbox{arctanh}\left(\sqrt{\frac{\varphi-1}{\varphi}}\right) - 
\sqrt{\varphi-1} \right).
\end{equation}
For $\varphi \to \infty$ this has the asymptotic expression
\begin{equation}\label{GminBS}
\mathcal{G}^{(-)}_{\rm BS}(\varphi) =
\frac{1}{\sigma} \left( \sqrt{\varphi} \log(4\varphi) -
2 \sqrt{\varphi} + O(1)\right)\,,\quad \varphi\to \infty.
\end{equation}
We used here the asymptotic expansion of the arctanh function of large argument
\begin{equation}
\mbox{arctanh}(1-x) = - \frac12 \log(x/2) + O(1)\,, \quad x \to 0_+.
\end{equation}
The asymptotic expansion (\ref{GminBS}) shows that 
$(\mathcal{G}^{(-)}_{\rm BS}(\varphi))^2$
grows faster than linear as $\varphi \to \infty$. This ensures the existence
of a finite solution $x_*$ to the equation (\ref{xeq}) in the BS model.
These results hold also in the general local volatility model with local
volatility function $\sigma(S) \leq \overline{\sigma}$ bounded from above, 
as this implies a lower bound on the function 
$(\mathcal{G}^{(-)}(\varphi))^2$ which has the same growth 
properties as $\varphi\to \infty$ as those obtained above in the BS model. 

ii) $K < S_0$. 
Using the representation (\ref{J2}) for the rate function,
we have from Lemma~\ref{lemma:2}
that $\mathcal{I}(K,S_0)$ is a decreasing function of $K$ for 
$0 < K < S_0$.
The function $(\mathcal{G}^{(+)}(\chi))^2$ satisfies the technical conditions
for $f(x)$ in Lemma~\ref{lemma:2}: positivity follows from its
definition, and the decreasing  property for $0 < \chi < 1$, follows from 
the positivity of the integral $\mathcal{F}^{(+)}(\chi)$ defined in (\ref{F2def}).

The applicability of Lemma~\ref{lemma:2} requires also that the infimum 
in (\ref{F2def}) is not reached at the lower boundary. 
We start by considering first the Black-Scholes model, where the function 
$\mathcal{G}^{(+)}(\chi)$ can be computed exactly and is given by
\begin{equation}
\mathcal{G}^{(+)}_{\rm BS}(\chi) = \frac{1}{\sigma} 
\int_\chi^1 \frac{1}{z} \sqrt{z-\chi}dz
= \frac{2}{\sigma} \left(
\sqrt{1-\chi} - 
\sqrt{\chi} \mbox{arctan}\left(\sqrt{\frac{1}{\chi} - 1}\right) 
\right).
\end{equation}
The function $\mathcal{F}^{(+)}_{\rm BS}(\chi)$ is given by
\begin{equation}
\mathcal{F}^{(+)}_{\rm BS}(\chi) = - 2 \frac{d}{d\chi} \mathcal{G}^{(+)}_{\rm BS}(\chi)
= \frac{2}{\sigma \sqrt{\chi}} 
\mbox{arctan} \sqrt{\frac{1}{\chi}-1}\,.
\end{equation}

This has the asymptotic expression as $\chi \to 0$
\begin{equation}\label{GplBS}
\mathcal{F}^{(+)}_{\rm BS}(\chi) = \frac{\pi}{\sigma\sqrt{\chi}} 
+ O(1)\,,\quad \chi\to 0\,.
\end{equation}
Thus we have
\begin{equation}\label{Gplusinfty}
\lim_{\chi \to 0} \frac{d}{d\chi} (\mathcal{G}^{(+)}_{\rm BS}(\chi))^2  = - \infty,
\end{equation}
where we used $\lim_{\chi\to 0} \mathcal{G}^{(+)}_{\rm BS}(\chi) = \frac{2}{\sigma} <\infty$.

This implies that the function $(\mathcal{G}^{(+)}_{\rm BS}(\chi))^2$
satisfies the conditions (\ref{lemma2cond}) required for $f(x)$ in 
Lemma~\ref{lemma:2}, which are sufficient to ensure that the infimum is not 
reached at the lower boundary  $a$. 
We conclude that the statement of Lemma~\ref{lemma:2} applies to the BS model.
These results hold also in the general local volatility model with local
volatility function $\underline{\sigma} \leq \sigma(S) \leq \overline{\sigma}$ 
bounded from below and above. The upper bound on $\sigma(S)$ 
implies a lower bound for $\mathcal{F}^{(+)}(\chi)$, and the lower bound on 
$\sigma(S)$ implies an upper bound for 
$\mathcal{G}^{(+)}(0)$. Taken together these conditions 
ensure that the result (\ref{Gplusinfty}) holds also for the general local 
volatility model.

\end{proof}

\begin{lemma}\label{lemma:1}
Let $f(x)$ be a positive and increasing function $f(x)>0, f'(x)>0$, 
and define
\begin{equation}\label{Fdef}
F(z) = \inf_{x > z} \left( \frac{f(x)}{x-z} \right)\,.
\end{equation}
Furthermore, assume that $f(x)$ grows faster than $x$ as $x\to \infty$.
Then $F(z)$ is a positive and increasing function 
\begin{equation}
F(z) > 0 \,,\qquad F'(z) >0 \,.
\end{equation}
\end{lemma}

\begin{proof}
The condition for the extremum of the function appearing in the definition
of $F(z)$ is
\begin{equation}
\frac{f'(x)}{x-z} - \frac{f(x)}{(x-z)^2} = 0 \, .
\end{equation}
Denote the solution of this equation $x_*$. This is given explicitly as
\begin{equation}\label{xeq}
x_* - z = \frac{f(x_*)}{f'(x_*)} \,.
\end{equation}
This equation will have a solution for $z < x_* < \infty$ if $f(x)$ grows 
faster than linear as $x\to \infty$. This ensures that the infimum in
(\ref{Fdef}) is not reached at $x\to \infty$.

Substituting into (\ref{Fdef}) we get
\begin{equation}
F(z) = f'(x_*)  > 0\,,
\end{equation}
where we used $f'(x_*)>0$. 
In order to prove the $F'(z)>0$ property, we compute the derivative with 
respect to $z$
\begin{equation}\label{Fp}
F'(z) = f''(x_*) \frac{dx_*}{dz} = 
f''(x_*) \frac{(f'(x_*))^2}{f''(x_*) f(x_*)} > 0\,.
\end{equation}
The second equality is obtained by taking a derivative of (\ref{xeq}) 
with respect to $z$
\begin{eqnarray}
\frac{dx_*}{dz} - 1 &=& \frac{d}{dz} \frac{f(x_*)}{f'(x_*)}= 
\frac{d}{dx}\left( \frac{f(x_*)}{f'(x_*)}\right) 
\frac{dx_*}{dz} \\
&=& \frac{(f'(x_*))^2 - f(x_*) f''(x_*)}{(f'(x_*))^2}\frac{dx_*}{dz} \nonumber
\end{eqnarray}
which gives
\begin{equation}
\frac{dx_*}{dz} = \frac{(f'(x_*))^2}{f''(x_*) f(x_*)} \,.
\end{equation}
The inequality (\ref{Fp}) proves that $F(z)$ is a monotonically increasing 
function.

Note that we do not need a convexity condition on $f(x)$ to obtain the
monotonicity of $F(z)$. If $f''(x)>0$ then we get additionally that 
$\frac{dx_*}{dz}>0$, but this is not required for the monotonicity of $F(z)$.
\end{proof}

\begin{lemma}\label{lemma:2}
Let $f(x):[a,b] \to \mathbb{R}$ be a positive and decreasing function 
$f(x)>0, f'(x) < 0$, and define
\begin{equation}\label{Fdef2}
F(z) = \inf_{x < z} \left( \frac{f(x)}{z - x} \right)\,.
\end{equation}
Assuming that the infimum is not achieved on the boundary at $x=a$,
then $F(z)$ is a positive and decreasing function 
\begin{equation}
F(z) > 0 \,,\qquad F'(z) < 0 \,,\qquad a < z \leq b\,.
\end{equation}
\end{lemma}

\begin{proof}
The condition for the extremum of the function appearing in the definition
of $F(z)$ is
\begin{equation}
\frac{f'(x)}{z-x} + \frac{f(x)}{(z-x)^2} = 0 \, .
\end{equation}
Denote the solution of this equation $x_*$. This is given explicitly as
\begin{equation}\label{xeq2}
x_* - z = \frac{f(x_*)}{f'(x_*)} \,.
\end{equation}
We will assume that this equation has a solution for $a< x_* < z$. 
One possible way to ensure that the infimum is not reached at $a$ is to
require
\begin{equation}\label{lemma2cond}
\lim_{x\to a} f(x) < \infty\,,\quad
\lim_{x\to a} f'(x) = -\infty\,.
\end{equation}

Substituting into (\ref{Fdef2}) we get
\begin{equation}
F(z) = -f'(x_*)
\end{equation}
Let us prove the $F'(z)<0$ property. We have
\begin{equation}\label{Fp2}
F'(z) = - f''(x_*) \frac{dx_*}{dz} = -
f''(x_*) \frac{(f'(x_*))^2}{f''(x_*) f(x_*)} < 0\,.
\end{equation}
The second equality is obtained by taking a derivative of (\ref{xeq2}) with 
respect to $z$, which gives the same result as in the proof of 
Lemma~\ref{lemma:1}
\begin{equation}
\frac{dx_*}{dz} = \frac{(f'(x_*))^2}{f''(x_*) f(x_*)} \,.
\end{equation}
The inequality (\ref{Fp2}) proves that $F(z)$ is
a monotonically decreasing function.
\end{proof}

%%%%%%%%%%%%%%%%%%%%%%%%%%%%%%%%%%%%%%%%%%%%%%%%%%%%%%%%

\begin{proof}[Proof of Proposition \ref{RateFunctionBS}]

The variational problem in Proposition~\ref{VarProp} simplifies for 
the case of a constant local volatility function $\sigma(S) = \sigma$, 
corresponding to the Black-Scholes model. 
The rate function is given by
\begin{equation}\label{Jcaldef}
\mathcal{J}_{\rm BS}(K/S_0) = 
\sigma^2 \mathcal{I}_{\rm BS}(K,S_0)=\frac12 \int_0^1 [f'(t)]^2 dt,
\end{equation}
where $f(t)$ is the solution of the Euler-Lagrange equation (the constant $a$
is related as $a=\lambda\sigma^2$ to the Lagrage multiplier appearing in the 
proof of Proposition~\ref{VarProp})
\begin{equation}\label{ELeq}
f''(t) = a e^{f(t)},
\end{equation}
with boundary conditions
\begin{equation}\label{BC}
f(0) = 0\,,\qquad f'(1)=0\,.
\end{equation}
The unknown constant $a$ is determined by the condition
\begin{equation}\label{const}
\int_0^1 dt e^{f(t)} = \frac{K}{S_0}\,.
\end{equation}

The solution of the differential equation (\ref{ELeq}) can be found exactly.
Two independent solutions of this equation are
\begin{align}\label{g1sol}
& f_1(x) = \beta x - 2 \log \left(
\frac{e^{\beta x} +\gamma}{1 + \gamma} \right), \\
\label{g2sol}
& f_2(x) = - 2 \log \left| \cos( \xi x + \eta) \right| + 
2 \log \left| \cos\eta \right| \,.
\end{align}
The solution $f_1(x)$ of the equation (\ref{ELeq}) was given 
in \cite{GuasoniRobertson},
where the same equation appears in the context of an optimal importance sampling 
problem for Asian options in continuous time in the Black-Scholes model. 

The parameters in these functions are determined by the boundary conditions 
(\ref{BC}), and by requiring that the functions satisfy the Euler-Lagrange
equation (\ref{ELeq}). 

It is easy to check, by direct substitution of (\ref{g1sol}), (\ref{g2sol}) 
into (\ref{ELeq}), that these functions are solutions of the Euler-Lagrange 
equation. 
The boundary condition at $x=0$ is satisfied automatically.
Matching the constant factor in the Euler-Lagrange equation and
requiring that also the boundary condition at $x=1$ is satisfied, gives 
the following constraints.

For $f_1(x)$ we have the equations
\begin{align}\label{eq1}
& 2\gamma \beta^2 = - a (1+\gamma)^{2},\\
\label{eq2}
&  \gamma = e^{\beta} \,,
\end{align}
which determine $a,\gamma$ in terms of $\beta$. 

For $f_2(x)$ we have
\begin{align}\label{eq1m}
& 2\xi^2 = a \cos^2| \eta|, \\
\label{eq2m}
& 2\xi \tan|\xi + \eta| = 0,
\end{align}
which give $\eta = -\xi + k\pi$ with $k\in\mathbb{Z}$, 
and determine $a$ in terms of $\xi$. 
The multiple solutions for $\eta$ give the
same solution $f_2(x)$ as we have $\cos(\xi(x-1)+k\pi) = (-1)^k 
\cos(\xi(x-1))$.

The constants $\xi,\beta$ can be determined from the constraint (\ref{const}).
For the solution (\ref{g1sol}), this is
\begin{equation}\label{eq3}
\int_0^1 e^{f_1(x)} dx = \frac{1 + \gamma}{e^{\beta} + \gamma} 
\frac{e^\beta - 1}{\beta} = \frac{1}{\beta} \sinh\beta = \frac{K}{S_0},
\end{equation}
which reproduces (\ref{betaeq}). This is an equation for $\beta$ which 
has solutions only for $K>S_0$.

For the solution (\ref{g2sol}) we obtain
\begin{equation}\label{eq3m}
\int_0^1 e^{f_2(x)} dx= \int_0^1 \frac{\cos^2\xi}{\cos^2(\xi(x-1))}dx
= \frac{1}{\xi} \cos^2\xi  \tan\xi = \frac{K}{S_0} \,,\qquad 
0\leq \xi < \frac{\pi}{2}\,.
\end{equation}
For $|\xi|\geq\frac{\pi}{2}$ the integral is divergent. 
This reproduces (\ref{lambdaeq}), and gives an equation for $\xi$ which 
has solutions only for $K<S_0$. The solution must satisfy $\xi\in(-\frac{\pi}{2},\frac{\pi}{2})$.
Note that if $\xi^{\ast}$ is a solution to \eqref{lambdaeq} in $(0,\frac{\pi}{2})$, then $-\xi^{\ast}\in(-\frac{\pi}{2},0)$ is also
a solution to \eqref{lambdaeq} which yields the same value of $\mathcal{J}(K/S_{0})$.
Therefore, without loss of generality, we can assume that $\xi\in(0,\frac{\pi}{2})$ which
has a unique solution to \eqref{lambdaeq}.

In conclusion, the solution of the Euler-Lagrange equation (\ref{ELeq}) 
with the boundary conditions (\ref{BC}) and constraint (\ref{const}) is
\begin{equation}\label{fsol}
f(x)=
\begin{cases}
\beta x - 2 \log \left(
\frac{e^{\beta x} + e^\beta}{1 + e^\beta} \right) & \, K \geq S_0 \\
\log\left(\frac{\cos^2\xi}{\cos^2 (\xi(x-1))}\right) & \, K \leq S_0
\end{cases}\,,
\end{equation}
where $\beta$ and $\xi$ are the solutions of the equations (\ref{betaeq}) 
and (\ref{lambdaeq}), respectively.

Finally, the rate function $\mathcal{J}_{\rm BS}(K/S_0)$ is found by 
substituting the solution (\ref{fsol}) for $f(x)$ into the equation 
(\ref{Jcaldef}) and performing the integration.
This reproduces the result (\ref{JBSresult}), which concludes the proof of 
Proposition~\ref{RateFunctionBS}.
\end{proof}

\begin{proof}[Proof of Proposition \ref{prop:asympt}]
(1) $K\geq S_0$. For $K/S_0 \to \infty$, the solution of the equation (\ref{betaeq})
approaches $\beta \to \infty$. This suggests writing this equation as
\begin{equation}
e^\beta = (2\beta) \frac{K}{S_0} \frac{1}{1-e^{-2\beta}},
\end{equation}
or
\begin{equation}\label{betaiter}
\beta = x + \log (2\beta) - \log ( 1 - e^{-2\beta}) 
      = x + \log (2\beta) + O(e^{-2\beta}),
\end{equation}
with $x=\log(K/S_0)$. This implies that $\beta=x + O(\log x)$ as $x\rightarrow\infty$, and
we can improve this estimate by iteration of (\ref{betaiter}),
starting with the first order approximation
\begin{equation}
\beta^{(0)} = x + o(x) .
\end{equation}
By substitution into (\ref{betaiter}) we get the successive iterations
\begin{eqnarray}
\beta^{(1)} &=& x + \log (2x) + O(e^{-2x}), \\
\beta^{(2)} &=& x + \log (2x) + \log [2 (x + \log(2x))] + O(e^{-2x}) \\
&=& x + \log(2x) + \log(2x) + \log\left( 1 + \frac{\log(2x)}{x}\right) + O(e^{-2x}) \nonumber \\
&=& x + 2 \log(2x) + \frac{\log(2x)}{x} + O(x^{-2})\,. \nonumber 
\end{eqnarray}
The asymptotic expansion of the rate function is obtained by substituting into 
(\ref{JBSresult}), and expanding to the order shown.
This gives the result (\ref{JBSasympt}).

(2) $K \leq S_0$. The parameter $\xi$ is obtained by solving the equation 
(\ref{lambdaeq}). It is clear that as $K/S_0\to 0$,
we have $\xi \to \pi/2$. 

It is convenient to introduce $\zeta$ defined as $\xi = \frac{\pi}{2}
 - \zeta$ with $\zeta \to 0$ in the $K/S_0\to 0$ limit. This is given by
the solution of the equation
\begin{equation}
\sin(2\zeta) = \frac{K}{S_0} (\pi - 2\zeta)\,.
\end{equation}
This equation can be solved again by iteration starting with $\zeta^{(0)}=0$.

The first two iterations are
\begin{eqnarray}
\zeta^{(1)} &=& \frac{\pi}{2} \frac{K}{S_0}\, , \\
\zeta^{(2)} &=& \frac{\pi}{2} \frac{K}{S_0} 
  \left(1 - \frac{K}{S_0}\right) + O((K/S_0)^3)\,.
\end{eqnarray}
Finally, substituting into (\ref{JBSresult}) we obtain (\ref{JBSasympt}).
\end{proof}

\begin{proof}[Proof of Proposition \ref{LVexpRF}]
We present the proof on the case $K \geq S_0$, the case $K<S_0$ is treated 
in a similar way, and leads to the same final result (\ref{RFexp}). 

For given log-strike $x=\log(K/S_0) \geq 0$, one has to find $f_1$ by 
solving the equation (\ref{eqf1}), and the rate function is obtained by 
substituting this value for $f_1$ into (\ref{Iresult}). We note that as 
$x \downarrow 0$, we have $f_1\downarrow 0$. Therefore it is 
reasonable to look for a solution for $f_1$ by expanding in $x$ around the 
ATM point $x=0$.

It is convenient to introduce the auxiliary variable $z_1$ such that
$f_1 = Y(z_1)$, with $Y(z) := Z^{-1}(z)$ the inverse of the function 
$Z(y) = \int_0^y \frac{dw}{\sigma(S_0e^w)}$. From the definition it is 
clear that for $x$ small, $z_1$ is also small, and they both go to zero 
simultaneously.
The rationale of introducing $z_1$ is to absorb the dependence on $\sigma(S)$
in $F^{(-)}(f_1), G^{(-)}(f_1)$ into a change of the integration variable.

The strategy of the proof will be:

Step 1. Express the equation (\ref{eqf1}) for $f_1$ as an equation for $z_1$,
expanded to a given order in $z_1$.

Step 2. Invert the expansion obtained in Step 1 and derive an expansion for $z_1$
in powers of the log-strike $x$ to a given order. 

Step 3. Express the result (\ref{Iresult}) for the rate function ${\mathcal I}(K,S_0)$
for $K\geq S_0$ as a series in $z_1$. Further, use here the expansion for 
$z_1$ in powers of $x$ obtained in Step 2. 

{\bf Step 1.}
We start by deriving the expansion of the functions $G^{(-)}(f_1)$ and 
$F^{(-)}(f_1)$ in powers of $z_1 = Z(f_1)$. 
First we absorb the dependence on the local volatility function $\sigma(S)$ in
the definitions of these functions into a new integration variable $z$ defined
such that $dz = dy/\sigma(S_0e^y)$. This gives
\begin{eqnarray}
G^{(-)}(f_1) &=& \int_0^{f_1} \frac{1}{\sigma(S_0 e^y)} \sqrt{e^{f_1} - e^{y}} dy=
                 \int_0^{z_1}  \sqrt{e^{Y(z_1)} - e^{Y(z)}} dz \, ,\\
F^{(-)}(f_1) &=& \int_0^{f_1} \frac{1}{\sigma(S_0 e^y)} 
                 \frac{1}{\sqrt{e^{f_1} - e^{y}}} dy=
                 \int_0^{z_1} \frac{dz}{\sqrt{e^{Y(z_1)} - e^{Y(z)}}} \,,
\end{eqnarray}
where $Y(z)$ is the inverse of the function $Z(y)= \int_0^y
\frac{dw}{\sigma(S_0e^w)}$, and $z_1 = Z(f_1)$.

Substituting the Taylor series of $Y(z)$ in (\ref{Ydef}), 
expanding in $z, z_1$ to a given order, and evaluating the integrals, we get
\begin{eqnarray}
G^{(-)}(f_1) &=& \sqrt{b_1} z_1^{3/2} \bigg(
\frac23 + \frac{7}{30b_1} (b_1^2+2b_2) z_1 \\
& & +
\frac{1}{1680 b_1^2}(81b_1^4+628 b_1^2 b_2 - 284 b_2^2 + 912 b_1 b_3) z_1^2
+ O(z_1^3)
\bigg) \, ,\nonumber \\
F^{(-)}(f_1) &=& \frac{1}{\sqrt{b_1}} \sqrt{z_1} \bigg(
2 - \frac{5}{6b_1} (b_1^2+2b_2) z_1 \\
& & 
+ \frac{1}{240b_1^2} (41 b_1^4 - 12 b_1^2 b_2 + 516 b_2^2 - 528 b_1 b_3)
z_1^2 + O(z_1^3) \bigg) \,.\nonumber
\end{eqnarray}

Next we express the equation (\ref{eqf1}) giving $f_1$ in terms of the 
log-strike $x=\log(K/S_0)$, as an expansion in $z_1$. This is
\begin{equation}\label{eqf1p}
x = \log\left( e^{f_1} - \frac{G^{(-)}(f_1)}{F^{(-)}(f_1)} \right),
\end{equation}
Recalling that $f_1 = Y(z_1) = b_1 z_1 + b_2 z_1^2 + \cdots $, we obtain
by substitution on the right-hand side and Taylor expansion in powers of $z_1$
\begin{eqnarray}\label{xseries}
x &=& \log \left( e^{Y(z_1)} - \frac{G^{(-)}(f_1)}{F^{(-)}(f_1)} \right) \\
&=&  \frac23 b_1 z_1 + \frac{1}{45}( b_1^2 + 22 b_2) z_1^2 \nonumber \\
& & + 
\frac{1}{2835 b_1}
(b_1^4+150 b_1^2 b_2 + 48 b_2^2 + 1026 b_1 b_3) z_1^3 + O(z_1^4) \,. \nonumber
\end{eqnarray}

{\bf Step 2.}
Next we invert the series (\ref{xseries})
to get $z_1$ as an expansion in $x$
\begin{eqnarray}
z_1 &=& \frac{3}{2b_1} x + \left(-\frac{3}{40 b_1} - \frac{33b_2}{20b_1^3} \right) x^2 \\
&+& \frac{1}{1400 b_1^5} (8b_1^4 + 87 b_1^2 b_2 + 4962 b_2^2 - 2565 b_1 b_3) x^3 
+ O(x^4) \,. \nonumber
\end{eqnarray}

{\bf Step 3.}
Finally, we insert this expansion into the expression for the rate function 
(\ref{Iresult}), which gives the expansion of the rate function in powers of $x$
\begin{eqnarray}
\mathcal{I}(K,S_0) &=& \frac12 F^{(-)}(f_1) G^{(-)}(f_1) 
= \frac{3}{2b_1^2} x^2 - \frac{3}{10b_1^4} (b_1^2+12b_2) x^3   \\
& & +
\frac{1}{1400 b_1^6} (109 b_1^4 + 936 b_1^2 b_2 + 14976 b_2^2 - 6480 b_1 b_3)
x^4 + O(x^5)\,.\nonumber
\end{eqnarray}
This reproduces the result (\ref{RFexp}).
\end{proof}

%%%%%%%%%%%%%%%%%%%%%%%%%%%%%%%%%%%%%%%%%%%%%%%%%%%%

\begin{proposition}\label{InfiniteSigma}
For any fixed $K,S_{0},T$, the price of an Asian option in the Black-Scholes
model approaches the following value in the infinite volatility limit
\begin{equation*}
\lim_{\sigma\rightarrow\infty}C_{BS}(K,S_{0},\sigma,T)=e^{-rT}\frac{1}{T}\int_{0}^{T}S_{0}e^{(r-q)t}dt.
\end{equation*}
\end{proposition}

\begin{proof}
For any $\epsilon>0$, 
\begin{align*}
&\left|C_{BS}(K,S_{0},\sigma,T)-e^{-rT}\mathbb{E}\left[\left(\frac{1}{T}\int_{\epsilon}^{T}S_{0}e^{(r-q)t+\sigma W_{t}-\frac{1}{2}\sigma^{2}t}dt
-K\right)^{+}\right]\right|
\\
&\leq e^{-rT}\mathbb{E}\left[\frac{1}{T}\int_{0}^{\epsilon}S_{0}e^{(r-q)t+\sigma W_{t}-\frac{1}{2}\sigma^{2}t}dt\right]
=e^{-rT}\frac{1}{T}\int_{0}^{\epsilon}S_{0}e^{(r-q)t}dt,
\end{align*}
where the last term is independent of $\sigma$ and it goes to $0$ as $\epsilon\rightarrow 0$.
On the other hand, for almost every sample path, the Brownian motion $W_{t}$ is continuous
in $t$ and therefore $\sup_{0\leq t\leq T}W_{t}<\infty$. It follows that
\begin{equation*}
\frac{1}{T}\int_{\epsilon}^{T}S_{0}e^{(r-q)t+\sigma W_{t}-\frac{1}{2}\sigma^{2}t}dt
\leq e^{\sigma\sup_{0\leq t\leq T}W_{t}-\frac{1}{2}\sigma^{2}\epsilon}\frac{1}{T}\int_{\epsilon}^{T}S_{0}e^{(r-q)t}dt
\rightarrow 0,
\end{equation*}
a.s. as $\sigma\rightarrow\infty$. By bounded convergence theorem, 
\begin{equation*}
\lim_{\sigma\rightarrow\infty}\mathbb{E}\left[\left(K-\frac{1}{T}\int_{\epsilon}^{T}S_{0}e^{(r-q)t+\sigma W_{t}-\frac{1}{2}\sigma^{2}t}dt
\right)^{+}\right]=K. 
\end{equation*}
>From put-call parity (Remark~\ref{PCParity}),
\begin{equation*}
\lim_{\sigma\rightarrow\infty}\mathbb{E}\left[\left(\frac{1}{T}\int_{\epsilon}^{T}S_{0}e^{(r-q)t+\sigma W_{t}-\frac{1}{2}\sigma^{2}t}dt
-K\right)^{+}\right]=\frac{1}{T}\int_{\epsilon}^{T}S_{0}e^{(r-q)t}dt.
\end{equation*}
Finally, we let $\epsilon\rightarrow 0$ to complete the proof.
\end{proof}

%%%%%%%%%%%%%%%%%%%%%%%%%%%%%%%%%%%%%%%%%%%%%%%%%%%%%

\begin{proof}[Proof of Proposition \ref{prop:18}]
(i) 
For the Black-Scholes model with $r=q=0$, the stock price follows a geometric 
Brownian motion, i.e., $S_{t}=S_{0}e^{\sigma W_{t}-\frac{1}{2}\sigma^{2}t}$. 
The price of the Asian call option in the Black-Scholes model is
\begin{equation*}
C_A^{(BS)}(S_{0},K,\sigma,T)=
\mathbb{E}\left[\left(S_{0}\int_{0}^{1}e^{\sigma W_{tT}-\frac{1}{2}\sigma^{2}tT}dt-K\right)^{+}\right].
\end{equation*}
>From the Brownian scaling property, 
\begin{equation*}
C_A^{(BS)}(S_{0},K,\sigma,T)=
\mathbb{E}\left[\left(S_{0}\int_{0}^{1}
e^{\sigma\sqrt{T}(W_{tT}/\sqrt{T})-\frac{1}{2}(\sigma^{2}T)t}dt-K\right)^{+}\right]
:= \tilde C_A^{(BS)}(S_{0},K,\sigma^{2}T)
\end{equation*}
can be viewed as a function of $\sigma^{2}T$. 
Moreover the result in Carr et al. \cite{CarrVega} implies that 
$\tilde C_A^{(BS)}(S_{0},K,\sigma^{2}T)$ is increasing as a function of 
$\sigma^{2}T$.

Consider first an out-of-the-money Asian call option $K>S_0$ in the
local volatility model (\ref{dynamics}) and denote its price
$C_A(S_0,K,T)$. We have $\tilde C_A^{(BS)}(S_{0},K,0)=0$ and 
$\tilde C_A^{(BS)}(S_{0},K,\sigma^{2}T)>0$ for any $\sigma^{2}T>0$.
We have 
$C_A(S_{0},K,T)=
\tilde C_A^{(BS)}(S_{0},K,\sigma_{\text{implied}}^{2}T)\rightarrow 0$
as $T\rightarrow 0$. Therefore, 
$\lim_{T\rightarrow 0}\sigma_{\text{implied}}^{2}(S_{0},K,T)T=0$.
Hence, by Theorem~\ref{MainThm} and Proposition~\ref{RateFunctionBS}, 
\begin{align*}
&\lim_{T\rightarrow 0}(\sigma_{\text{implied}}^{2}(S_{0},K,T)T)
\log \tilde C_A^{(BS)}(S_{0},K,\sigma_{\text{implied}}^{2}(S_{0},K,T)T)
=-\mathcal{J}_{BS}(K/S_{0})\,,
\\
&\lim_{T\rightarrow 0}T\log C_A(S_{0},K,T)=-\mathcal{I}(K,S_{0}).
\end{align*}
This implies that
\begin{align*}
\lim_{T\rightarrow 0}\sigma_{\text{implied}}^{2}(S_{0},K,T)
&=\lim_{T\rightarrow 0}
\frac{\sigma_{\text{implied}}^{2}(S_{0},K,T)T\log \tilde C_A^{(BS)}(S_{0},K,\sigma_{\text{implied}}^{2}(S_{0},K,T)T)}
{T\log C_A(S_{0},K,T)}
\\
&=\frac{\mathcal{J}_{BS}(K/S_{0})}{\mathcal{I}(K,S_{0})}.
\end{align*}
The same result is obtained for an out-of-the-money Asian put option $K<S_0$.
The argument follows through in a completely analogous way.

(ii) Next, let us consider the at-the-money Asian call option. 
>From Theorem \ref{ATMThm}, we have
$\lim_{T\rightarrow 0}\frac{1}{\sqrt{T}}C_A(K,S_{0},T)=
\frac{1}{\sqrt{6\pi}}\sigma(S_{0})S_{0}$.
Notice that $\tilde C_A^{(BS)}(K,S_{0},\sigma_{\text{implied}}^{2}(K,S_{0},T)T)
=C_A(K,S_{0},T)$
and $C_A^{(BS)}(K,S_{0},\sigma,T)=\tilde C_A^{(BS)}(K,S_{0},\sigma^{2}T)$ 
can be viewed as a function of $\sigma^{2}T$.
Thus, $\lim_{T\rightarrow 0}
\frac{\tilde C_A^{(BS)} (K,S_{0},\sigma_{\text{implied}}^{2}(K,S_{0},T)T)}
{\sqrt{6\pi}\sqrt{T}\sigma_{\text{implied}}(K,S_{0},T)}=S_{0}$. 
Therefore, $\lim_{T\rightarrow 0}\sigma_{\text{implied}}(K,S_{0},T)=
\sigma(S_{0})$. 
The same result is obtained for an at-the-money Asian put option.
\end{proof}

%%%%%%%%%%%%%%%%%%%%%%%%%%%%%%%%%%%%%%%%%%%%%%%%%%%%%%%%%

\begin{proof}[Proof of Proposition \ref{propSigBS}]
(i) Notice that when $r=q=0$, $A(T)=S_{0}$ and thanks to the Brownian scaling 
property, the price of an European option in the Black-Scholes model
$\tilde C_E^{(BS)}(K,S_0,\sigma^2 T) :=
\mathbb{E}\left[\left(S_{0}e^{\sigma W_{T}-\frac{1}{2}\sigma^{2}T}-K\right)^{+}\right]$
can be viewed as a function of $\sigma^{2}T$. This is a strictly increasing 
function of this argument. 
We have by definition of the equivalent log-normal volatility
$C_A(S_{0},K,T)=\tilde C_E^{(BS)}(S_{0},K,\Sigma_{\rm LN}^{2}T)$ where 
$C_A(S_0,K,T)$ denotes the price of an Asian call option in the local volatility
model (\ref{dynamics}).

We proceed analogously to the proof of Proposition \ref{prop:18}.
Consider first an out-of-the-money Asian call option $K>S_0$, for which we have 
$C_A(S_{0},K,T)=\tilde C_E^{(BS)}(S_{0},K,\Sigma_{\rm LN}^{2}T)\rightarrow 0$
as $T\rightarrow 0$. 
Therefore, $\lim_{T\rightarrow 0}\Sigma_{\rm LN}^{2}(S_{0},K,T)T=0$.
Hence,
\begin{align*}
&\lim_{T\rightarrow 0}(\Sigma_{\rm LN}^{2}(S_{0},K,T)T)
\log \tilde C_E^{(BS)}(S_{0},K,\Sigma_{\rm LN}^{2}(S_{0},K,T)T)
=-\frac12 \log^2(K/S_0)\,,
\\
&\lim_{T\rightarrow 0}T\log C_A(S_{0},K,T)=-\mathcal{I}(K,S_{0}).
\end{align*}
This implies that
\begin{align*}
\lim_{T\rightarrow 0}\Sigma_{\rm LN}^{2}(S_{0},K,T)
&=\lim_{T\rightarrow 0}
\frac{\Sigma_{\rm LN}^{2}(S_{0},K,T)T
\log \tilde C_E^{(BS)}(S_{0},K,\Sigma_{\rm LN}^{2}(S_{0},K,T)T)}
{T\log C_A(S_{0},K,T)}
\\
&=\frac{\log^2(K/S_{0})}{2\mathcal{I}(K,S_{0})}.
\end{align*}
The same result is obtained for an out-of-the-money Asian put option $K<S_0$.

The corresponding result for the equivalent normal volatility of the 
Asian option follows from the short maturity relation between the Black-Scholes
and Bachelier implied volatilities, see for example Corollary 2 in 
\cite{Grunspan}. 

(ii) We use again the fact that the price of an European option in the
Black-Scholes model depends only on $\sigma\sqrt{T}$ and denote it as before
$\tilde C_E^{(BS)}(K,S_{0},\sigma^{2}T)$.
For at-the-money case $S_{0}=K$, we have
$\lim_{T\rightarrow 0}
\frac{\tilde C_E^{(BS)}(K,S_{0},\sigma^{2}T)}
{\sigma\sqrt{2\pi}\sqrt{T}}=S_{0}$.
Therefore, we have
\begin{equation}
\lim_{T\rightarrow 0}\Sigma_{\rm LN}(K,S_{0},T)
=\lim_{T\rightarrow 0}\frac{\frac{C_A(K,S_{0},T)}{\sqrt{T}}}
 {\frac{\tilde C_E^{(BS)}(K,S_{0},\Sigma_{\rm LN}^{2}(K,S_{0},T)T)}
 {\Sigma_{\rm LN}(K,S_{0},T)\sqrt{T}}}
=\frac{\frac{1}{\sqrt{6\pi}}\sigma(S_{0})S_{0}}{\frac{1}{\sqrt{2\pi}}S_{0}}
=\frac{1}{\sqrt{3}}\sigma(S_{0}).
\end{equation}
For the Bachelier model with $K=S_{0}$, $\mathbb{E}[(S_{0}+\sigma W_{T}-K)^{+}]$
only depends on $\sigma\sqrt{T}$, 
and $\lim_{T\rightarrow 0}\frac{1}{\sigma\sqrt{T}}
\mathbb{E}[(S_{0}+\sigma W_{T}-K)^{+}]=\frac{1}{\sqrt{2\pi}}$. 
Therefore,
\begin{equation}
\lim_{T\rightarrow 0}\Sigma_{\rm N}(K,S_{0},T)
=\lim_{T\rightarrow 0}\frac{\frac{C_A(K,S_{0},T)}{\sqrt{T}}}{\frac{\mathbb{E}[(S_{0}+\Sigma_{\rm N}(K,S_{0},T)W_{T}-K)^{+}]}
{\Sigma_{\rm N}(K,S_{0},T)\sqrt{T}}}
=\frac{\frac{1}{\sqrt{6\pi}}\sigma(S_{0})S_{0}}{\frac{1}{\sqrt{2\pi}}}
=\frac{1}{\sqrt{3}}\sigma(S_{0})S_{0}.
\end{equation}

We note that these results can also be extracted from Theorem 5.1 of \cite{RR}.
\end{proof}

%%%%%%%%%%%%%%%%%%%%%%%%%%%%%%%%%%%%%%%%%%%%%%%%%%%%%%%%%%%%%

\subsection{Proofs of the Results in Section \ref{sec:floating}}

\begin{proof}[Proof of Proposition \ref{FloatingThm}]
(i) By an argument similar to that used in the proof of Theorem \ref{MainThm}, 
\begin{equation}
\lim_{T\rightarrow 0}T\log C_f(T)=\lim_{T\rightarrow 0}T\log\mathbb{P}\left(\kappa S_{T}\geq\int_{0}^{1}S_{tT}dt\right).
\end{equation}
Then, the result follows from the sample path large deviation of $\mathbb{P}(S_{t\cdot}\in\cdot)$ on $L_{\infty}[0,1]$ 
and the contraction principle. The asymptotic results for $P(T)$ follows from put-call parity.

(ii) Similar to (i).

(iii) Following the same arguments as in the proof of Theorem \ref{ATMThm}, we can show that
for $\kappa=1$, we have

\textbf{Claim 1}. As $T\rightarrow 0$,
\begin{equation*}
\left|\mathbb{E}\left[\left(e^{(r-q)T}X_{T}-\frac{1}{T}\int_{0}^{T}e^{(r-q)t}X_{t}dt\right)^{+}\right]
-\mathbb{E}\left[\left(X_{T}-\frac{1}{T}\int_{0}^{T}X_{t}dt\right)^{+}\right]\right|=O(T).
\end{equation*}

\textbf{Claim 2}. As $T\rightarrow 0$,
\begin{equation*}
\mathbb{E}\left[\max_{0\leq t\leq T}|X_{t}-\hat{X}_{t}|\right]=O(T).
\end{equation*}

\textbf{Claim 3}. As $T\rightarrow 0$, 
\begin{equation*}
\left|\mathbb{E}\left[\left(X_{T}-\frac{1}{T}\int_{0}^{T}X_{t}dt\right)^{+}\right]
-\mathbb{E}\left[\left(\hat{X}_{T}-\frac{1}{T}\int_{0}^{T}\hat{X}_{t}dt\right)^{+}\right]\right|=O(T).
\end{equation*}

And \textbf{Claim 1}, \textbf{Claim 2} and \textbf{Claim 3} imply that for $\kappa=1$,
\begin{equation}
\left|C_f(T)-\mathbb{E}\left[\left(\kappa\hat{X}_{T}-\frac{1}{T}\int_{0}^{T}\hat{X}_{t}dt\right)^{+}\right]\right|
=O(T),
\end{equation}
as $T\rightarrow 0$, where we recall that 
$\hat{X}_{t}=S_{0}+\sigma(S_{0})S_{0}W_{t}$. Since $\kappa=1$, 
$\kappa\hat{X}_{T}-\frac{1}{T}\int_{0}^{T}\hat{X}_{t}dt$ is a Gaussian random variable with mean zero and variance
\begin{align}
&\mathbb{E}\left[\left(\hat{X}_{T}-\frac{1}{T}\int_{0}^{T}\hat{X}_{t}dt\right)^{2}\right]
\\
&=\sigma(S_{0})^{2}S_{0}^{2}\mathbb{E}\left[\left(W_{T}-\frac{1}{T}\int_{0}^{T}W_{t}dt\right)^{2}\right]
\nonumber
\\
&=\sigma(S_{0})^{2}S_{0}^{2}\left[\mathbb{E}[W_{T}^{2}]+\frac{1}{T^{2}}\mathbb{E}\left[\left(\int_{0}^{T}W_{t}dt\right)^{2}\right]
-\frac{2}{T}\int_{0}^{T}\mathbb{E}[W_{T}W_{t}]dt\right]
\nonumber
\\
&=\sigma(S_{0})^{2}S_{0}^{2}\left[T+\frac{T}{3}-\frac{2}{T}\frac{T^{2}}{2}\right]=\sigma(S_{0})^{2}S_{0}^{2}\frac{T}{3}.
\nonumber
\end{align}
Hence, we proved the desired result for $C_f(T)$. The result for $P_f(T)$ 
is similar.
\end{proof}

\begin{proof}[Proof of Proposition \ref{VarFloating}]
Define a new optimizer function $g(t) = f(t) + \log S_0$, satisfying the
boundary condition $f(0)=0$ and the constraint $\int_0^1 dt e^{f(t)} = \kappa e^{f(1)}$. 
Proceeding in a similar fashion as in the proof of Proposition~\ref{VarProp},
the constraint is taken into account by introducing a Lagrange multiplier $\lambda$
and considering the variational problem associated with the auxiliary functional
\begin{equation}
\Lambda[f] = \frac12 \int_0^1 
\left( \frac{f'(t)}{\sigma(S_0 e^{f(t)})} \right)^2 dt+ \lambda
\left( \int_0^1 e^{f(t)} dt - \kappa e^{f(1)} \right)\,.
\end{equation}
By Lemma~\ref{Lemma1} the solutions of this variational problem satisfy the 
Euler-Lagrange equation (\ref{ELf}) and the transversality condition 
(\ref{BCf}) at $t=1$. 

To prove (\ref{Ifres}) we note that from Lemma~\ref{Lemma2} we have
\begin{equation}\label{Econserv}
\frac12 \left( \frac{f'(t)}{\sigma(S_0 e^{f(t)})}\right)^2 - \lambda e^{f(t)} = 
\frac12 \lambda^2 \kappa^2 e^{2f(1)} \sigma^2(S_0 e^{f(1)}) - \lambda e^{f(1)} \,.
\end{equation}
We substituted on the right hand side the transversality condition (\ref{BCf}).
Integrating this relation over $t:(0,1)$ gives the result (\ref{Ifres}).

Finally, we can eliminate $\lambda$ with the help of the relation (\ref{lambdarel}). This
is obtained by comparing two alternative expressions for $f'(0)$. First, by integrating the
Euler-Lagrange equation (\ref{ELf}) over $t:(0,1)$ gives
\begin{equation}
\frac{f'(0)}{\sigma(S_0)} = \lambda \left\{ \kappa e^{f(1)} \sigma(S_0 e^{f(1)}) - I_s[f] \right\}
\end{equation}
An alternative relation is obtained by taking $t=0$ in (\ref{Econserv}). Eliminating $f'(0)$ 
among these two equations gives (\ref{lambdarel}).
\end{proof}

\begin{proof}[Proof of Proposition \ref{BSfloating}]
The rate function for short-maturity asymptotics of floating strike Asian options 
in the Black-Scholes model is given by the variational problem
\begin{equation}
\mathcal{I}_f(\kappa) = \inf_f \frac{1}{2\sigma^2} \int_0^1 [f'(t)]^2 dt,
\end{equation}
where $f(0) = 0$, $f\in\mathcal{AC}[0,1]$ and subject to the constraint
\begin{equation}
\int_0^1 e^{f(t)} dt= \kappa e^{f(1)}\,.
\end{equation}

This can be solved by introducing a Lagrange multiplier $\lambda$ and
considering the auxiliary variational problem
\begin{equation}\label{Lambdadef1}
\Lambda[f] = \frac{1}{2\sigma^2} \int_0^1 [f'(t)]^2 dt+ \lambda \left(
\int_0^1 e^{f(t)} dt- \kappa e^{f(1)} \right)\,,
\end{equation}
over all functions satisfying $f(0)=0$. 

The solution of the variational problem (\ref{Lambdadef1}) satisfies the 
Euler-Lagrange equation 
\begin{equation}
f''(t) = \lambda \sigma^2 e^{f(t)} ,
\end{equation}
which must be solved with the boundary conditions
\begin{align*}
&(\mbox{BC1}):\quad f(0) = 0 , \\
&(\mbox{BC2}):\quad f'(1) = \lambda \sigma^2 \kappa e^{f(1)} , \\
&(\mbox{BC3}):\quad f'(0) = 0 \,.
\end{align*}
The boundary condition (BC2) is a transversality condition. The condition (BC3) is 
new, and follows from the relation of the variational problem for $\Lambda[f]$
to an equivalent variational problem which is obtained by defining a new
function $h$ as
\begin{equation}
f(t) = h(1-t) + f(1)\,.
\end{equation}

Expressed in terms of $h$, the rate function $\mathcal{I}_f(\kappa)$ is given by
\begin{equation}\label{Ifh}
\mathcal{I}_f(\kappa) = \inf_h \frac{1}{2\sigma^2} \int_0^1 [h'(t)]^2 dt
\end{equation}
where $h(0) = 0$, $h\in\mathcal{AC}[0,1]$ and it is subject to the constraint
\begin{equation}
\int_0^1 e^{h(t)} dt= \kappa \,.
\end{equation}
This is identical to the variational problem for the rate function 
$\mathcal{I}(\kappa S_0,S_0)$ for fixed strike Asian options (\ref{Idef}), in the limit 
of constant volatility $\sigma(S)=\sigma$. The solution of this variational 
problem is presented in Proposition~\ref{RateFunctionBS}, and is expressed 
in terms of the Black-Scholes rate function $\mathcal{J}_{\rm BS}(\kappa)$ as 
shown in (\ref{IfBS}).

For completeness, we give a complete solution of the variational problem (\ref{Ifh}), 
and along the way prove the additional boundary condition (BC3).
Introducing again a Lagrange multiplier $\eta$, the variational problem (\ref{Ifh})
can be transformed into an unconstrained optimization of the functional
\begin{equation}
\Xi[h]= \frac{1}{2\sigma^2} \int_0^1 [h'(t)]^2 dt+ \eta \left(
\int_0^1 e^{h(t)} dt- \kappa \right)\,,
\end{equation}
over all functions satisfying $h(0)=0$. 
The solution of this variational problem satisfies the Euler-Lagrange equation
\begin{equation}
h''(t) = \eta \sigma^2 e^{h(t)} ,
\end{equation}
with the boundary conditions
\begin{equation}
h(0) = 0 \,, \qquad
h'(1) = 0  \,.
\end{equation}
The second boundary condition is a transversality condition.

However, $f$ and $h$ are related, so the Euler-Lagrange equations and 
boundary conditions for $h$ must imply also equivalent conditions on $f$.
Comparing the Euler-Lagrange equations for $f$ and $h$ gives
\begin{equation}
\eta = \lambda e^{f(1)} = \lambda e^{-h(1)}\,.
\end{equation}
Using $f'(t) = - h'(1-t)$, the transversality condition $h'(1) = 0$
gives 
\begin{equation}
f'(0) = - h'(1) = 0\,.
\end{equation}
This proves the boundary condition (BC3) on $f'(0)$ quoted above.

Using the additional boundary condition (BC3), it is possible to give a simple
expression for the Lagrange multiplier $\lambda$
\begin{equation}\label{lambdasol}
\lambda = \frac{2}{\sigma^2 \kappa^2} \frac{e^{f(1)} - 1}{e^{2f(1)}}.
\end{equation}

This follows from the conservation of the quantity
\begin{equation}
\frac{1}{2\sigma^2} [f'(t)]^2 - \lambda \sigma^2 e^{f(t)}.
\end{equation}
This gives
\begin{equation}
- \lambda \sigma^2 = \frac12 \lambda^2 \sigma^2 \kappa^2 e^{2f(1)}- 
\lambda \sigma^2 e^{f(1)},
\end{equation}
where we used $f'(0)$ and the boundary condition (BC2) at $t=1$. This 
concludes the proof of the relation (\ref{lambdasol}). 

The solution for $h(1)$ can be found explicitly as shown in the proof of 
Proposition~\ref{RateFunctionBS}. This is given in equation (\ref{fsol}), from 
which we find
\begin{equation}
e^{h(1)}=
\begin{cases}
\cosh^2(\frac12 \beta) & \, \kappa > 1, \\
\cos^2 \lambda & \, \kappa < 1, \\
\end{cases}
\end{equation}
with $\beta, \lambda$ the solutions of the equations (\ref{betaeq}) and (\ref{lambdaeq}),
respectively (with the substitutions $K/S_0 \mapsto \kappa$).
\end{proof}

\section*{Acknowledgements}

The authors are grateful to the Editor, an Associate Editor and two anonymous referees
for their helpful suggestions that greatly improved the quality of the paper.
The authors would like to thank Tai-Ho Wang for helpful discussions, and for
bringing to our attention work in progress on a related topic \cite{THW}.
Lingjiong Zhu is partially supported by NSF Grant DMS-1613164.

\end{document}